\documentclass[10pt,a4paper]{article}

\usepackage[utf8]{inputenc}
\usepackage[T1]{fontenc}
\usepackage[english]{babel}
\usepackage{footmisc}
\usepackage{amsmath}
\usepackage{amsfonts}
\usepackage{amssymb}
\usepackage{amsthm}

\newtheorem{thm}{Theorem}
\newtheorem{lem}{Lemma}
\newtheorem{cor}{Corollary}
\newtheorem{propo}{Proposition}
\newtheoremstyle{problem}{}{}{}{0pt}{}{}{0pt}{}
\theoremstyle{problem}
\newtheorem*{prb}{}

\usepackage{multicol}
\usepackage{graphicx}
\usepackage{enumitem}

\author{Andreas Darmann and Janosch Döcker}

\title{On simplified NP-complete variants of \\ \textsc{Not-All-Equal 3-Sat} and \textsc{3-Sat}}

\begin{document}
\maketitle

\begin{abstract}
We consider simplified, monotone versions of \textsc{Not-All-Equal 3-Sat} and \textsc{3-Sat}, variants of the famous \textsc{Satisfiability Problem} where each clause is made up of exactly three distinct literals. We show that \textsc{Not-All-Equal 3-Sat} remains NP-complete even if (1) each variable appears exactly four times, (2) there are no negations in the formula, and (3) the formula is linear, i.e., each pair of distinct clauses shares at most one variable. 

Concerning \textsc{3-Sat} we prove several hardness results for monotone formulas with respect to a variety of restrictions imposed on the variable appearances.  \textsc{Monotone 3-Sat} is the restriction of \textsc{3-Sat} to monotone formulas, i.e. to formulas in which each clause contains only unnegated variables or only negated variables, respectively. In particular, we show that, for any $k\geq 5$, \textsc{Monotone 3-Sat} is NP-complete even if each variable appears exactly $k$ times unnegated and exactly once negated. In addition, we show that \textsc{Monotone 3-Sat} is NP-complete even if each variable appears exactly three times unnegated and three times negated, respectively. In fact, we provide a complete analysis of \textsc{Monotone 3-Sat} with exactly six  appearances per variable. Further, we prove that the problem remains NP-complete when restricted to instances in which each variable appears either exactly once unnegated and three times negated or the other way around. Thereby, we improve on a result by Darmann et al.~\cite{Darmann2018} showing NP-completeness for four appearances per variable. Our stronger result also implies that  \textsc{3-Sat} remains NP-complete even if each variable appears exactly  three times unnegated and once negated, therewith complementing a result by Berman et al.~\cite{Berman2003}. 
\end{abstract}

\section{Introduction}

The famous Boolean satisfiability problem, and in particular \textsc{$3$-Satisfiability}, can be considered \textit{the}  classical decision problem in computer science. \textsc{$3$-Satisfiability} has been the first problem shown to be \textsc{NP}-complete decades ago (Cook~\cite{cook71}) and is of undisputed theoretical and practical importance; it both appears in practical applications of routing, scheduling and artificial intelligence (see, e.g., Devlin and O'Sullivan~\cite{devlin08}, Nam et al.~\cite{Nam99}, Horbach et al.~\cite{Horbach12}, and Kautz and Selman~\cite{kautz96}), and is the most prominent problem, and probably the most frequently used one, for complexity analysis of decision problems. Therefore, it has continuously attracted researchers through decades focusing on the computational complexity of variants of the satisfiability problem (for recent work see, e.g., Pilz~\cite{pilz19} or Paulusma and Szeider~\cite{Paul2019}). 

In this paper, we add to that branch of literature and investigate the computational complexity\footnote{We assume the reader to be familiar with the basic concepts of the theory of NP-completeness and refer to Garey and Johnson~\cite{garey79} for an extensive introduction.} of restricted variants of \textsc{$3$-Satisfiability} and its variation \textsc{Not-All-Equal $3$-Satisfiability}, which is also known to be \textsc{NP}-complete (Schaefer~\cite{schaefer78}). 

In \textsc{$3$-Satisfiability}, we are given a set of propositional variables and a collection of clauses, where each clause contains three literals. The question is whether there is a satisfying truth assignment, i.e., whether we can satisfy all clauses by assigning truth values to the variables. In what follows, we will refer to \textsc{3-Sat} as the version of \textsc{$3$-Satisfiability} in which each clause is made up of three \emph{distinct} literals -- which is the setting we focus on in this paper -- and to \textsc{3-Sat*} as the version in which the three literals of a clause are not necessarily distinct. \textsc{Not-All-Equal Satisfiability} imposes an additional restriction on truth assignments by asking whether  there is a truth assignment such that for each clause at least one literal evaluates to true and at least one to false, respectively. As a consequence of Schaefer's dichotomy theorem~\cite{schaefer78} \textsc{Not-All-Equal Satisfiability} is NP-complete even if each clause is made up of three literals. In addition, Dehghan et al.~\cite[pp.\,1362f.]{Dehghan2015} show that \textsc{Not-All-Equal Satisfiability} remains NP-complete even if each variable appears unnegated exactly three times (i.e., there are no negations present at all), and each clause is a disjunction of either two or three distinct variables. In contrast, \textsc{Not-All-Equal 3-Sat} (all clauses have exactly three distinct variables) can be solved in polynomial time in case that there are no negations present and each variable appears at most three times (Porschen et al.~\cite[Theorem 4]{porschen04}, see also discussion in \cite[p.\,186]{porschen04}).

The first main focus of this paper is laid on the monotone variant of \textsc{Not-All-Equal 3-Sat}. According to the common convention an instance of \textsc{Not-All-Equal 3-Sat} is \emph{monotone}\footnote{We point out that \textit{monotonicity} has different meanings for \textsc{3-Satisfiability} and \textsc{Not-All-Equal 3-Satisfiability}, which is certainly not ideal but appears to be the established notation.\label{fn:monotone}} if and only if negations are completely absent, i.e., there are no negated variables in the formula. Porschen et al.~\cite{porschen14} studied variants of \textsc{Not-All-Equal 3-Sat} that restrict the interleaving of clauses; e.g., in \emph{linear} formulas each pair of distinct clauses shares at most one variable. In particular, Porschen et al.~\cite[Theorem 3]{porschen14} show that \textsc{Monotone Not-All-Equal 3-Sat} is NP-complete even for linear instances. In this paper we improve upon the result of Porschen et al.~\cite{porschen14} by showing that \textsc{Monotone Not-All-Equal 3-Sat} remains NP-complete for linear instances in which each variable appears exactly four times. Therewith, we also improve upon a result of Karpinski and Piecuch~\cite{karppiec18}, who show that \textsc{Not-All-Equal 3-Sat*} (possibly with duplicates of literals in the same clause) is NP-complete even if each variable appears at most $4$ times in the formula.  

The second main focus of this paper is laid on \textsc{Monotone 3-Sat} where each clause is monotone\footref{fn:monotone}, i.e., contains only unnegated or only negated variables, respectively. It is known that \textsc{Monotone 3-Sat} is NP-complete~\cite{gold78,li97}, and that intractability holds even if (1) each variable appears exactly 4 times~\cite[Corollary 4]{Darmann2018}. We show that this problem remains NP-complete even if condition~(1) is replaced by either one of the following four conditions: 
\begin{itemize}
\item (1a) each variable appears exactly $k$ times unnegated and $k$ times negated, respectively, for every fixed integer $k \geq 3$,
\item (1b) each variable appears exactly $k'$ times unnegated and once negated, respectively, for every fixed integer $k' \geq 5$,
\item (1c) each variable appears exactly $p$ times unnegated and $6-p$ times negated, respectively, for all $p \in \{1, 2, 3, 4, 5\}$, or
\item (1d) each variable appears exactly three times unnegated and once negated or three times negated and once unnegated. 
\end{itemize}
We remark that the hardness result for condition (1d) improves upon the result for condition (1) by Darmann et al.~\cite[Corollary 4]{Darmann2018}. Also, as a by-product, we derive the result that the classical \textsc{3-Sat} problem remains NP-complete  even if each variable appears exactly  three times unnegated and once negated (observe that this  implies hardness also for the vice versa case where each variable appears exactly once unnegated and three times negated). Therewith, we complement results of Tovey~\cite{tovey84} and Berman et al.~\cite{Berman2003}: The former showed that \textsc{3-Sat} remains NP-complete even if each variable appears in at most 4 clauses  and it is trivial if the number of variable appearances is bounded by 3~\cite[Theorem 2.3 and Theorem 2.4]{tovey84}; Berman et al.~\cite[Theorem 1]{Berman2003} added to that result by showing that NP-completeness holds even if each variable appears exactly twice negated and twice unnegated, respectively. 

Further related literature is concerned with the planar\footnote{In that respect, planarity refers to the corresponding graph property of the following associated bipartite graph: there is a vertex for each variable $v$ and for each clause $c$, and an edge connects a variable vertex $v$ with a clause vertex $c$ if and only if variable $v$ appears in clause $c$.} variants of \textsc{(Monotone) 3-Satisfiability}.   Both \textsc{Planar 3-Satisfiability} and \textsc{Planar Monotone 3-Satisfiability} are known to be NP-complete even in restricted settings (e.g., see~\cite{lichtenstein82,kratochvil94} respectively \cite{deberg12,Darmann2018}), while Pilz~\cite[Theorem 11]{pilz19} shows that all instances of \textsc{Planar Monotone 3-Sat}, i.e., where each clause contains three \textit{distinct} variables, are satisfiable. Moreover, the planar variant of \textsc{Not-All-Equal 3-Sat} can be solved in polynomial time~\cite{moret88}.

The paper is structured as follows. In Section~\ref{sec:prel} we introduce basic notation and formally state the considered decision problems. The focus of Section~\ref{sec:nae} is laid on restricted instances of \textsc{Not-All-Equal 3-Sat}, and in Section~\ref{sec:3sat} we provide hardness results for \textsc{Monotone 3-Sat} in restricted settings. Finally, Section~\ref{sec:con} concludes the paper with a concise summary of the results and challenges for future research.  

\section{Preliminaries}\label{sec:prel}

Let $V = \{x_1, \ldots, x_n\}$ be a set of propositional \emph{variables}. For the remainder of the paper we simply say variable instead of propositional variable since all variables take on values in $\{T, F\}$, where $T$ represents true and $F$ false, respectively. A \emph{literal} is a variable or its negation, i.e., an element of $L_V = \{x_i, \overline{x_i} \mid x_i \in V\}$. A \emph{clause} is a subset of $L_V$, and a $k$-clause contains exactly $k$ distinct literals. Further, a clause is \emph{monotone} if either all contained variables are negated or none of them is. In the setting of nae-satisfiability, which we define later, a clause is monotone if none of the contained variables is negated, i.e., if the clause is a subset of $V$. A Boolean formula $C$ in \emph{conjunctive normal form} (CNF) is a collection of $m$ clauses, i.e., $C = \bigcup_{j=1}^m \{c_j\}$. It is also common to use logical connectives, e.g. $\vee$ and $\wedge$, to describe a Boolean formula. Then, $C$ is a conjunction of clause $\bigwedge_{j = 1}^m c_j$, where $c_j = (\ell_{j,1} \vee \ell_{j,2} \vee \ldots \vee \ell_{j,i_j})$ is a disjunction of literals. We use the set notation to emphasize that we do not allow duplicates of literals in clauses. For one result, where we allow duplicates, we describe a clause by a multiset instead. For instance, $\{x_1, x_1, x_3\}$ represents a clause in this setting that contains $x_1$ twice. We denote the total number of appearances of a variable $x_i \in V$ in a formula $C$ by $a(x_i)$. A Boolean formula is \emph{linear} if all pairs of distinct clauses share at most one variable. A \emph{truth assignment} is a mapping $\beta \colon V \rightarrow \{T, F\}$ which extends to literals in the obvious way, i.e., for $\ell=x_i$ we have $\beta(\ell)=\beta(x_i)$ and for $\ell=\overline{x_i}$ we have  $\beta(\ell) \in \{T, F\} \setminus \beta(x_i)$, $i \in \{1,2, \ldots, n\}$. A clause $c_j$ is \emph{satisfied} under $\beta$ if $\beta(\ell) = T$ for at least one $\ell \in c_j$. Further, $c_j$ is \emph{nae-satisfied} if there are literals $\ell, \ell' \in c_j$ such that $\beta(\ell) \neq \beta(\ell')$. A Boolean formula $C = \bigcup_{j=1}^m \{c_j\}$ in CNF is \emph{satisfiable} (\emph{nae-satisfiable}) if there exists a truth assignment $\beta \colon V \rightarrow \{T, F\}$ such that all clauses $c_1, \ldots, c_m$ are satisfied (nae-satisfied). We say that a truth assignment $\beta'$ for $V'$ \emph{extends} a truth assignment $\beta$ for $V$ if $V \subseteq V'$ and $\beta'(v) = \beta(v)$ for all $v \in V$. \\

\subsection{Problem statements}

The decision problems considered in this work are stated below; we abbreviate {\textsc{Not-All-Equal 3-Sat}} with {\textsc{NAE-3-Sat}}.

\begin{prb}
\noindent{\sc Monotone NAE-3-Sat-E4} \\
{\bf Instance.} A set $V$ of variables, and a collection $C$ of clauses over $V$ such that each clause $c \in C$ contains $|c| = 3$ distinct variables, every variable appears in exactly four clauses and there is no negation in the formula. \\
{\bf Question.} Is there a truth assignment for $V$ such that each clause in $C$ has at least one true literal and at least one false literal? \end{prb}

\begin{prb}
\noindent{\sc Monotone 3-Sat-$(p, q)$} \\
{\bf Instance.} A set $V$ of variables, and a collection $C$ of clauses over $V$ such that each clause $c \in C$ contains $|c| = 3$ distinct variables, either all or none of them negated, and every variable appears unnegated in exactly $p$ clauses and negated in exactly $q$ clauses, respectively.\\
{\bf Question.} Is there a truth assignment for $V$ such that each clause in $C$ has at least one true literal? \end{prb}

\begin{prb}
\noindent{\sc Monotone 3-Sat*-$(2, 2)$} \\
{\bf Instance.} A set $V$ of variables, and a collection $C$ of clauses over $V$ such that each clause $c \in C$ is a multiset containing $|c| = 3$ variables, either all or none of them negated, and every variable appears exactly twice negated and twice unnegated, respectively.\\
{\bf Question.} Is there a truth assignment for $V$ such that each clause in $C$ has at least one true literal? \end{prb}

\begin{prb}
\noindent{\sc Monotone 3-Sat-E4} \\
{\bf Instance.} A set $V$ of variables, and a collection $C$ of clauses over $V$ such that each clause $c \in C$ contains $|c| = 3$ distinct variables, either all or none of them negated, and every variable appears in exactly four clauses. \\
{\bf Question.} Is there a truth assignment for $V$ such that each clause in $C$ has at least one true literal? \end{prb}

Finally, dropping the monotonicity requirement, we will consider the following restricted variant of \textsc{3-Sat}. 

\begin{prb}
\noindent{\sc 3-Sat-E4} \\
{\bf Instance.} A set $V$ of variables, and a collection $C$ of clauses over $V$ such that each clause $c \in C$ contains $|c| = 3$ distinct variables, and every variable appears in exactly four clauses. \\
{\bf Question.} Is there a truth assignment for $V$ such that each clause in $C$ has at least one true literal? \end{prb}

Note that all of the above decision problems belong to the class NP. Hence, the NP-completeness proofs in this paper reduce to showing NP-hardness of the respective problem. 

\section{A simplified variant of \textsc{Not-All-Equal 3-Sat}}\label{sec:nae}

We begin our study with {\sc Monotone NAE-3-Sat-E4}, proving its NP-completeness in Section~\ref{sub:MonNAE}. This result, in turn, is then used in Section~\ref{sub:MonNAE-linear} to derive the even stronger result that {\sc Monotone NAE-3-Sat-E4} remains NP-complete even when restricted to linear formulas. 

\subsection{Hardness of \textsc{Monotone NAE-3-Sat-E4}}\label{sub:MonNAE}

For our first result, NP-completeness of {\sc Monotone NAE-3-Sat-E4}, we give two different proofs. The reason for doing so is that the first proof has the advantage of being relatively simple, while featuring the drawback of using an auxiliary gadget to increase the number of variable appearances; the latter is avoided in the second proof. 

\begin{thm}\label{thm:MonNAE_3SATE4}
{\sc Monotone NAE-3-Sat-E4} is NP-complete.
\end{thm}

\noindent\textit{Proof 1 of Theorem~\ref{thm:MonNAE_3SATE4}.}
We show NP-hardness of {\sc Monotone NAE-3-Sat-E4} by reduction from {\sc Monotone NAE-3-Sat} (see, e.g., Porschen et al.~\cite[Theorem 3]{porschen14} for a proof that the latter problem is NP-complete). Let $\mathcal{I} = (V, C)$ be an instance of {\sc Monotone NAE-3-Sat}. Let $n := |V|$ denote the number of variables, $m := |C|$ the number of clauses and recall that $a(x_i)$ denotes the number of appearances of a variable $x_i \in V$ in the formula $C$. Further, let the set of variables be given as $V := \{x_1,x_2,\ldots,x_n\}$. 

For each variable $x_i$, we replace the $j$th appearance with a new variable~$x_{i,j}$ and introduce the clauses
\[
\operatorname{EQ}(x_{i,a(x_i)}, x_{i,1}) \cup \bigcup_{j = 1}^{a(x_i) - 1} \operatorname{EQ}(x_{i,j}, x_{i,j+1}),
\]
where $\operatorname{EQ}(x_{i,s}, x_{i,t})$ is an \emph{equality gadget} (a set of clauses) enforcing that $x_{i,s}$ and $x_{i,t}$ are mapped to the same truth value by any satisfying assignment. More precisely, a truth assignment $\beta$ for $\{x_s, x_t\}$ can be extended to a truth assignment $\beta'$ for all variables appearing in $\operatorname{EQ}(x_{i,s}, x_{i,t})$ that nae-satisfies $\operatorname{EQ}(x_{i,s}, x_{i,t})$ if and only if $\beta(x_{i,s}) = \beta(x_{i,t})$. We construct this gadget in two steps. First, we define a \emph{non-equality gadget} enforcing that two variables are set to different truth values in any nae-satisfying truth assignment.
 
Consider the set of clauses   
\[
\operatorname{NE}(x, y) := \{\{x,y,a\}, \{x,y,b\}, \{a,b,u\}, \{a,b,v\}, \{a,b,w\}, \{u,v,w\}\},
\]
where $a,b,u,v,w$ are new variables not appearing anywhere else, e.g., the clause sets $\operatorname{NE}(x, y)$ and $\operatorname{NE}(y, z)$ do not have any common variables except of~$y$. In order to nae-satisfy the last clause in $\operatorname{NE}(x, y)$, at least one of $u,v,w$ is set to true and at least one of them is set to false. Hence, by construction of the three preceding clauses, $a$ and $b$ are set to different truth values. Then, due to the first two clauses $x$ and $y$ are set to different truth values in any truth assignment that nae-satisfies $\operatorname{NE}(x, y)$. Now, the equality gadget is defined as
\[
\operatorname{EQ}(x,y) := \operatorname{NE}(p, q) \cup \operatorname{NE}(p, r) \cup \{\{x,q,r\}, \{y,q,r\}\},
\]
where $p$, $q$ and $r$ are new variables not appearing anywhere else. Note that by construction of the two non-equality gadgets, $q$ and $r$ are set to the same truth value. Hence, due to the two last clauses, $x$ and $y$ are set to the same truth value. By symmetry of nae-satisfying truth assignments, we can, thus, extend any truth assignment $\beta$ for $\{x, y\}$ with $\beta(x) = \beta(y)$ to a truth assignment that nae-satisfies $\operatorname{EQ}(x,y)$.  

Note that each variable $x_{i,j}$ appears in two equality gadgets, once in each gadget, and in exactly one clause of the original instance. Moreover, each introduced variable appears in at most four clauses. With the following gadget, we can increase the appearances of a variable by one, while only introducing variables with exactly four appearances. Let 
\[
\operatorname{P1}(x) := \{\{x,a,b\}, \{a,c,d\}, \{a,b,e\}, \{a,d,e\}, \{b,c,d\}, \{b,c,e\}, \{c,d,e\}\},
\]
where $a,b,c,d,e$ are new variables not appearing anywhere else. Note that these clauses are satisfiable independently of the truth value of $x$ by setting each variable in $\{a,c,e\}$ true and each variable in $\{b,d\}$ false. Now, we can use this gadget to increase the appearances of each variable until it appears exactly four times. The number of introduced variables and clauses is clearly polynomial and the verification of the reduction is straightforward.  \hfill $\square$\\

We now present a second proof for Theorem~\ref{thm:MonNAE_3SATE4} which reduces from the more general  \textsc{NAE-{3-Sat*}} problem and does not require a separate gadget to increase the number of  variable appearances. The proof will make use of the two following lemmata. 

\begin{lem}\label{lem:NE-gadget} Let $\operatorname{NE}(x, y)$ be the following set of clauses, where $V_{\text{aux}} = \{a, b, \ldots, f\}$ are new variables.  

\begin{multicols}{3}
\begin{enumerate}
\item $\{x, a, b\}$
\item $\{y, c, d\}$
\item $\{y, e, f\}$
\item $\{c, e, f\}$
\item $\{b, c, e\}$
\item $\{a, c, f\}$
\item $\{a, d, e\}$
\item $\{a, b, d\}$
\item $\{b, d, f\}$
\end{enumerate}
\end{multicols}

Then, a truth assignment $\beta$ for $\{x, y\}$ can be extended to a truth assignment $\beta'$ for $\{x, y\} \cup V_{\text{aux}}$ that nae-satisfies $\operatorname{NE}(x, y)$ if and only if $\beta(x) \neq \beta(y)$.
\end{lem}

\begin{proof}
First, we can nae-satisfy all clauses in $\operatorname{NE}(x, y)$ by setting all variables in $\{x, c, d, e\}$ true (resp. false) and all variables in $\{y, a, b, f\}$ false (resp. true). 
Second, assume towards a contradiction that there is an nae-satisfying assignment $\beta$ with $\beta(x) = \beta(y) = T$. We consider all four possible assignments of the variables $a$ and $c$ to truth values.

\noindent\underline{Case $\beta(a) = F, \beta(c) = F$:} By clause $6$ we have $\beta(f) = T$. Then, by clause $3$ we have $\beta(e) = F$. By clauses $5$ and $7$ we have $\beta(b) = T$ and $\beta(d) = T$, respectively. Hence, all literals in clause $9$ evaluate to true, i.e., $\beta$ does not nae-satisfy clause~$9$.  

\noindent\underline{Case $\beta(a) = F, \beta(c) = T$:} By clause $2$ we have $\beta(d) = F$. Then, by clauses $7$ and $8$ we have $\beta(e) = T$ and $\beta(b) = T$, respectively. By clause $5$ we have $\beta(b) = F$. Thus, we have $\beta(b) \neq \beta(b)$, a contradiction.

\noindent\underline{Case $\beta(a) = T, \beta(c) = F$:} By clause $1$ we have $\beta(b) = F$. Then, by clause $5$ we have $\beta(e) = T$. By clause $7$ we have $\beta(d) = F$. Then, by clause $9$ we have $\beta(f) = T$. Hence, all literals in clause $3$ evaluate to true, i.e., $\beta$ does not nae-satisfy clause $3$.   

\noindent\underline{Case $\beta(a) = T, \beta(c) = T$:} By clauses $1$, $2$ and $6$ we have $\beta(b) = F$, $\beta(d) = F$ and $\beta(f) = F$, respectively. Hence, all literals in clause $9$ evaluate to false, i.e., $\beta$ does not nae-satisfy clause $9$.

By symmetry of nae-satisfying truth assignments, there is no nae-satisfying assignment $\beta$ with $\beta(x) = \beta(y) = F$.
\end{proof}

\begin{lem}\label{lem:EQ-gadget} Let $\operatorname{EQ}(x, y)$ be the following set of clauses, where $V_{\text{aux}} = \{a, b, \ldots, i\}$ are new variables. 

\begin{multicols}{5}
\begin{enumerate}
\item $\{x, a, b\}$
\item $\{y, c, d\}$
\item $\{y, e, f\}$
\item $\{a, c, g\}$
\item $\{a, e, d\}$
\item $\{a, h, i\}$
\item $\{b, e, h\}$
\item $\{b, f, h\}$
\item $\{b, g, i\}$
\item $\{c, e, i\}$
\item $\{c, f, g\}$
\item $\{d, g, h\}$
\item $\{d, f, i\}$
\end{enumerate}
\end{multicols}

Then, a truth assignment $\beta$ for $\{x, y\}$ can be extended to a truth assignment $\beta'$ for $\{x, y\} \cup V_{\text{aux}}$ that nae-satisfies $\operatorname{EQ}(x, y)$ if and only if $\beta(x) = \beta(y)$.
\end{lem}

\begin{proof}
First, we can nae-satisfy all clauses in $\operatorname{EQ}(x, y)$ by setting all variables in $\{x, y, e, g, h, i\}$ true (resp. false) and all variables in $\{a,b,c,d,f\}$ false (resp. true). Hence, we can extend a truth assignment $\beta$ for $\{x, y\}$ to a truth assignment $\beta'$ for $\{x, y\} \cup V_{\text{aux}}$ that nae-satisfies $\operatorname{EQ}(x, y)$ if $\beta(x) = \beta(y)$. 

Second, assume towards a contradiction that $\beta(x) = F$ and $\beta(y) = T$ for a truth assignment $\beta$ that nae-satisfies $\operatorname{EQ}(x, y)$.

\noindent\underline{Case $\beta(a) = F, \beta(c) = F$:} By clauses $1$ and $4$ we have $\beta(b) = T$ and $\beta(g) = T$, respectively. Then, by clause $9$ we have $\beta(i) = F$. By clauses $6$ and $10$ we have $\beta(h) = T$ and $\beta(e) = T$, respectively. But then all literals in clause $7$ evaluate to true, i.e., $\beta$ does not nae-satisfy clause $7$, a contradiction to our assumption.

\noindent\underline{Case $\beta(a) = F, \beta(c) = T$:} By clauses $1$ and $2$ we have $\beta(b) = T$ and $\beta(d) = F$, respectively. Then, by clause $5$ we have $\beta(e) = T$. By clauses $7$ and $10$ we have $\beta(h) = F$ and $\beta(i) = F$, respectively. Therewith all literals in clause $6$ evaluate to false and hence $\beta$ does not nae-satisfy clause $6$, a contradiction.

\noindent\underline{Case $\beta(a) = T, \beta(c) = F$:}
\begin{itemize}
\item Case $\beta(e) = F$: By clause $10$ we have $\beta(i) = T$. Then, by clause $6$ we have $\beta(h) = F$. By clause $7$ we have $\beta(b) = T$. Then, by clause $9$ we have $\beta(g) = F$. By clause $12$ we have $\beta(d) = T$. Then, by clause $13$ we have $\beta(f) = F$. This, however, implies that $\beta$  does not nae-satisfy clause $11$, a contradiction.
\item Case $\beta(e) = T$: By clauses $3$ and $5$ we have $\beta(f) = F$ and $\beta(d) = F$, respectively. Then, by clauses $13$ and $11$ we have $\beta(i) = T$ and $\beta(g) = T$, respectively. By clauses $9$ and $6$ we have $\beta(b) = F$ and $\beta(h) = F$, respectively. Hence, $\beta$  does not nae-satisfy clause $8$, in contradiction with our assumption.
\end{itemize}
\noindent\underline{Case $\beta(a) = T, \beta(c) = T$:} By clauses $2$ and $4$ we have $\beta(d) = F$ and $\beta(g) = F$, respectively. Then, by clause $12$ we have $\beta(h) = T$. By clause $6$ we have $\beta(i) = F$. Then, by clauses $13$ and $9$ we have $\beta(f) = T$ and $\beta(b) = T$, respectively. Thus, $\beta$  does not nae-satisfy clause $8$, a contradiction.

Hence, there is no truth assignment $\beta$ with $\beta(x) = F$ and $\beta(y) = T$ that nae-satisfies $\operatorname{EQ}(x, y)$. By symmetry of nae-satisfying truth assignments, there is also no truth assignment $\beta$ with $\beta(x) = T$ and $\beta(y) = F$ that can be extended to a truth assignment that nae-satisfies $\operatorname{EQ}(x, y)$. 
\end{proof}

Now, we have the tools we need for our second proof of Theorem~\ref{thm:MonNAE_3SATE4}, i.e., that {\sc Monotone NAE-3-Sat-E4} is NP-complete.\\

\noindent\textit{Proof 2 of Theorem~\ref{thm:MonNAE_3SATE4}.}
We show NP-hardness by reduction from \textsc{NAE-3-Sat*}. NP-completeness of \textsc{NAE-3-Sat*} was established by Schaefer~\cite{schaefer78}. Let $\mathcal{I} = (V, C)$ be an instance of \textsc{NAE-3-Sat*}. Let $n := |X|$ denote the number of variables, $m := |C|$ the number of clauses and recall that $a(x_i)$ denotes the number of appearances of a variable $x_i \in V$ in the formula $C$. Further, let the set of variables be given as $V := \{x_1,x_2,\ldots,x_n\}$. 

For each variable $x_i \in V$, we replace the $j$th appearance with a new variable~$x_{i,j}$, such that $x_{i,j}$ is unnegated for $j \leq u(x_i)$ and negated for $j > u(x_i)$, where $u(x_i) \in \{0,1,\ldots,a(x_i)\}$ is the number of unnegated appearances of $x_i$ in~$C$. First, we make sure that, for each $x_i \in V$, all variables in $\{x_{i,j} \mid j \leq u(x_i)\}$ are mapped to the same truth value in any nae-satisfying assignment by introducing the clauses
\[
\bigcup_{i = 1}^n \bigcup_{j = 1}^{u(x_i)-1} EQ(x_{i,j}, x_{i,j+1}),  
\]
where $EQ(x_{i,j}, x_{i,j+1})$ is the equality gadget defined in Lemma~\ref{lem:EQ-gadget}. Second, we do the same for the variables in $\{x_{i,j} \mid j > u(x_i)\}$, i.e., we introduce the clauses
\[
\bigcup_{i = 1}^n \bigcup_{j = u(x_i) + 1}^{a(x_i)-1} EQ(x_{i,j}, x_{i,j+1}).
\]
Now, we delete all negations and make sure that $x_{i,j}$ and $x_{i,j'}$ with $j \leq u(x_i)$ and $j' > u(x_i)$ are to be mapped to different truth values by introducing 
\[
\bigcup_{\substack{1 \leq i \leq n \\ u(x_i) \not \in \{ 0, a(x_i)\}}} NE(x_{i,u(x_i)}, x_{i,u(x_i)+1}), 
\]
where $NE(x_{i,j}, x_{i,j'})$ is the non-equality gadget defined in Lemma~\ref{lem:NE-gadget}. Next, in order to get the right number of variable appearances, we introduce for each $x_i$ that appears only negated or only unnegated the clauses $EQ(x_{i,a(x_i)}, x_{i,1})$ and for each variable $x_{i'}$ that appears both negated and unnegated we introduce the clauses $NE(x_{i',a(x_{i'})}, x_{i',1})$. Thus, for each variable $x_i$ we get the ring structure
\[
EQ(x_{i,1}, x_{i,2}) \cup EQ(x_{i,2}, x_{i,3}) \cup \ldots \cup EQ(x_{i,a(x_i)-1}, x_{i,a(x_i)}) \cup EQ(x_{i,a(x_i)}, x_{i,1}),
\]
if $x_i$ appears only negated or only unnegated, and we get the ring structure
\begin{align*}
&\bigcup_{j = 1}^{u(x_i)-1} EQ(x_{i,j}, x_{i,{j+1}}) \cup NE(x_{i,u(x_i)}, x_{i,u(x_i)+1}) \cup {}\\ 
&\bigcup_{j = u(x_i) + 1}^{a(x_i)-1} EQ(x_{i,j}, x_{i,{j+1}}) \cup NE(x_{i,a(x_i)}, x_{i,1}),
\end{align*}
otherwise. It is straightforward to verify that the resulting instance is nae-satisfiable if and only if $\mathcal{I}$ is nae-satisfiable. 

 Note that for $a(x_i)>1$ each variable $x_{i,j}$ appears exactly once as the first argument and exactly once as the second argument of a gadget (it is not important of which gadget) yielding three appearances of $x_{i,j}$. Observe that in the case $a(x_i)=1$ we introduce $EQ(x_{i,a(x_i)}, x_{i,1}) = EQ(x_{i,1}, x_{i,1})$ only, hence yielding three appearances of $x_{i,1}$ by means of that gadget. Since each $x_{i,j}$ also replaces exactly one appearance of $x_i$ in the clause set $C$, we get exactly four appearances of $x_{i,j}$ in the constructed instance. All other variables introduced by the gadgets (variables of the gadgets that are not arguments are always newly created, i.e., these variables are \emph{not} shared between gadgets) appear exactly four times by construction. Hence, the resulting instance is indeed an instance of {\sc Monotone NAE-3-Sat-E4}. We conclude the proof by remarking that the transformation is  polynomial. \hfill $\square$
 
\subsection{Hardness of \textsc{Monotone NAE-3-Sat-E4} for linear formulas}\label{sub:MonNAE-linear}

In this section, we strengthen our result from the previous section by showing that {\sc Monotone NAE-3-Sat-E4} remains NP-complete even when restricted to linear formulas. We begin by stating the following lemma.  

\begin{lem}\label{lem:Linear-gadget} Let $\operatorname{EQ}(x, y, z, u)$ be the following set of clauses, where $V_{\text{aux}} = \{a, b, \ldots, f\}$ are new variables. 

\begin{multicols}{4}
\begin{enumerate}
\item $\{x, a, e\}$
\item $\{x, b, d\}$
\item $\{x, c, f\}$
\item $\{y, a, b\}$
\item $\{y, c, e\}$
\item $\{y, d, f\}$
\item $\{z, a, f\}$
\item $\{z, c, d\}$
\item $\{z, u, b\}$
\item $\{u, a, c\}$
\item $\{u, d, e\}$
\item $\{b, e, f\}$
\end{enumerate}
\end{multicols}

Then, a truth assignment $\beta$ for $\{x, y, z, u\}$ can be extended to a truth assignment $\beta'$ for $\{x, y, z, u\} \cup V_{\text{aux}}$ that nae-satisfies $\operatorname{EQ}(x, y, z, u)$ if and only if $\beta(x) = \beta(y) = \beta(z) = \beta(u)$. In addition, the above set of clauses is linear if the variables $x, y, z, u$ are pairwise distinct.
\end{lem}

\begin{proof}
First, by setting all variables in $\{x, y, z, u, e\}$ true and all variables in $\{a,b,c,d,f\}$ false we can nae-satisfy all clauses in $\operatorname{EQ}(x, y, z, u)$. Further, by flipping the truth values for these sets, we obtain a nae-satisfying truth assignment where $x, y, z$ and $u$ are all set false. Second, we show that $\beta(x) = \beta(y) = \beta(z) = \beta(u)$ for each assignment $\beta$ that nae-satisfies $\operatorname{EQ}(x, y, z, u)$. Let $\beta$ be a nae-satisfying assignment. Assume towards a contradiction that $\beta(x) \neq \beta(y)$. By symmetry of nae-satisfying truth assignments, we may assume that $\beta(x) = F$ and $\beta(y) = T$. Then, $\beta$ nae-satisfies the first six clauses if and only if $\beta$ satisfies (not necessarily nae-satisfies) the following set of 2-clauses:
\[
\{\{a, e\}, \{b, d\}, \{c, f\}, \{\bar{a}, \bar{b}\}, \{\bar{c}, \bar{e}\}, \{\bar{d}, \bar{f}\}\}
\]
Now, using resolution we obtain clauses $\{\bar{b}, e\}, \{\bar{e}, f\}, \{\bar{f}, b\}$ which are satisfied if $\beta$ satisfies the above set of 2-clauses. Since the inferred clauses form a cyclic implication chain, we have $\beta(b) = \beta(e) = \beta(f)$. Thus, clause 12 is not nae-satisfied which is a contradiction to the assumption that $\beta$ nae-satisfies $\operatorname{EQ}(x, y, z, u)$. Hence, $\beta(x) = \beta(y)$ and, by symmetry of nae-satisfying truth assignments, we may assume that $\beta(x) = \beta(y) = F$. If $\beta(z) = \beta(u) = F$, we are done. Let us consider the three remaining cases:
\begin{itemize}
\item If $\beta(z) = \beta(u) = T$, then $\beta(b) = F$ by clause 9. By clauses 2 and 4, we have $\beta(d) = T$ and $\beta(a) = T$, respectively. Then, by clause 7 and 11, we have $\beta(f) = F$ and $\beta(e) = F$, respectively. Thus, clause 12 is not nae-satisfied. Again, this is a contradiction to the assumption that $\beta$ nae-satisfies $\operatorname{EQ}(x, y, z, u)$. 
\item If $\beta(z) = T$ and $\beta(u) = F$, then $\beta$ nae-satisfies clauses 2, 6, 7, 8, 10 and 11 if and only if $\beta$ satisfies (again, not necessarily nae-satisfies) the following set of 2-clauses: 
\[
\{\{b, d\}, \{d, f\}, \{\bar{a}, \bar{f}\}, \{\bar{c}, \bar{d}\}, \{a, c\}, \{d, e\}\}.
\]
Using resolution, we obtain clauses $\{\{\bar{f} , c\}, \{\bar{c} , f\}, \{\bar{c} , b\}, \{\bar{c} , e\}\}$ which are satisfied by $\beta$ since $\beta$ satisfies the above set of 2-clauses. Now, by the first two inferred clauses and clause 3 (recall that $\beta(x) = F$), we have $\beta(c) = \beta(f) = T$. Then, by the latter two inferred clauses, we have $\beta(b) = \beta(e) = T$. Thus, clause 12 is not nae-satisfied, a contradiction.
\item If $\beta(z) = F$ and $\beta(u) = T$, then $\beta$ nae-satisfies clauses 1, 2, 6, 8, 10 and 11 if and only if $\beta$ satisfies the following set of 2-clauses: 
\[
\{\{a, e\}, \{b, d\}, \{d, f\}, \{c, d\}, \{\bar{a}, \bar{c}\}, \{\bar{d}, \bar{e}\}\}.
\]
Using resolution, we obtain clauses $\{\{\bar{c},\, e\},\, \{\bar{e},\, c\},\, \{\bar{e},\, b\},\, \{\bar{e},\, f\}\}$ which leads to a contradiction in a similar way as in the previous case (i.e., $\beta$ does not nae-satisfy clause 12). 
\end{itemize}
Hence, we conclude that $\beta(x) = \beta(y) = \beta(z) = \beta(u)$ for each assignment $\beta$ that nae-satisfies $\operatorname{EQ}(x, y, z, u)$. A truth assignment $\beta$ for $\{x, y, z, u\}$ can, thus, be extended to a truth assignment $\beta'$ for $\{x, y, z, u\} \cup V_{\text{aux}}$ that nae-satisfies $\operatorname{EQ}(x, y, z, u)$ if and only if $\beta(x) = \beta(y) = \beta(z) = \beta(u)$.

By considering each pair of distinct clauses in $\operatorname{EQ}(x, y, z, u)$ it is easy to verify that the set of clauses is linear if the variables $x, y, z, u$ are pairwise distinct.
\end{proof}

\begin{thm}
{\sc Monotone NAE-3-Sat-E4} is NP-complete for linear formulas.
\end{thm}

\begin{proof}
We show NP-hardness by reduction from {\sc Monotone NAE-3-Sat-E4}, for which NP-hardness was established in Theorem~\ref{thm:MonNAE_3SATE4}. Let $\mathcal{I} = (V, C)$ be an instance of  {\sc Monotone NAE-3-Sat-E4}. Let $n := |V|$ denote the number of variables, $m := |C|$ the number of clauses and let the set of variables be given as $V := \{x_1,x_2,\ldots,x_n\}$. For each variable $x_i \in V$, we replace the $j$th appearance with a new variable~$x_{i,j}$. Then, we make sure that, for each $x_i \in V$, all variables in $\{x_{i,1}, x_{i,2}, x_{i,3}, x_{i,4}\}$ are mapped to the same truth value in any nae-satisfying truth assignment by introducing the clauses
\[
\bigcup_{i=1}^n \operatorname{EQ}(x_{i,1}, x_{i,2}, x_{i,3}, x_{i,4}), 
\]
where $\operatorname{EQ}(x_{i,1}, x_{i,2}, x_{i,3}, x_{i,4})$ is the equality gadget defined in Lemma~\ref{lem:Linear-gadget}. The gadgets do not share any variables, i.e., each instance of the equality gadget has its own newly created auxiliary variables. Note that each variable still appears exactly four times, once in the original clause set and three times in an equality gadget. Further, since the variables $x_{i,1}, x_{i,2}, x_{i,3}, x_{i,4}$ are pairwise distinct, the subformulas defined by the equality gadgets are linear (see Lemma~\ref{lem:Linear-gadget}). Observe that the clauses of the original instance are pairwise disjoint after the variable replacement and each of these clauses shares at most one variable with any clause introduced by the gadgets. Note that each clause, except clause 9, in the $i$th instance of the equality gadget contains at most one variable that appears in the original clause set, i.e., at most one variable $x_{i,j}$ with $1 \leq i \leq n$ and $1 \leq j \leq 4$. Even though clause 9 (see the clause set introduced in Lemma~\ref{lem:Linear-gadget}) contains two variables $x_{i,3}$ and $x_{i,4}$ that appear outside the gadget, there is no other clause that contains both of them (otherwise some clause of the given instance of {\sc Monotone NAE-3-Sat-E4} contains the variable $x_i$ twice, a contradiction). Hence, the constructed formula is linear. By Lemma~\ref{lem:Linear-gadget} it follows that the constructed instance is nae-satisfiable if and only if $\mathcal{I}$ is nae-satisfiable.

We conclude the proof by remarking that the transformation is polynomial.\end{proof}

\section{Simplified variants of \textsc{Monotone 3-Sat}}\label{sec:3sat}

In this section, the focus is laid on restricted variants of {\sc Monotone 3-Sat}. In Section~\ref{sub:balanced} we consider the case of \textit{balanced} variable appearances, where each variable appears  unnegated and negated equally often. 
In Section~\ref{sub:once-negated} {\sc Monotone 3-Sat} is analyzed restricted to instances in which each variable appears exactly once negated. Section~\ref{sub:six} deals with a full dichotomy result for {\sc Monotone 3-Sat} when each variable appears exactly six times. Finally, we consider {\sc Monotone 3-Sat} restricted to instances in which each variable appears either three times unnegated and once negated or  once unnegated and three times negated in Section~\ref{sub:a-restr}. 

\subsection{Balanced variable appearances}\label{sub:balanced}

Section~\ref{sub:balanced} is structured as follows. We begin with a simple corollary stating NP-completeness of {\sc Monotone 3-Sat-(4,4)}, even in a restricted setting. 
Then we turn to {\sc Monotone 3-Sat-(3,3)} and, by the use of several lemmata, show its NP-completeness, leading to our first main result in the section that {\sc Monotone 3-Sat-$(k,k)$} is intractable for any choice of $k\geq 3$. \newline
Finally, we turn to instances in which each variable appears exactly twice unnegated and exactly twice negated. We show that {\sc Monotone 3-Sat-(2,2)} is either trivial, i.e., each instance is satisfiable, or NP-complete. That is, in order to confirm NP-completeness, it would suffice to find an unsatisfiable instance of {\sc Monotone 3-Sat-(2,2)}. We conclude Section~\ref{sub:balanced}, however, with proving NP-completeness for the case that the literals in the 3-clauses are not necessarily distinct, i.e., NP-completeness of {\sc Monotone 3-Sat*-(2,2)}. 

\subsubsection{\textsc{Monotone 3-Sat-}$(k,k)$, for $k\geq 3$}

\begin{cor}
{\sc Monotone 3-Sat-(4,4)} is NP-complete, even if no pair of clauses has exactly two variables and more than one literal in common.
\end{cor}

\begin{proof}
This follows from the simple standard transformation from \textsc{Not-All-Equal 3-Sat} to \textsc{3-Sat}: Given an instance of {\sc Monotone NAE-3-Sat-E4} where the formula is linear, introduce for each clause $\{\ell_1, \ell_2, \ell_3\}$ a second clause $\{\neg \ell_1, \neg \ell_2, \neg \ell_3\}$. Note that the resulting formula has the desired properties. 
\end{proof}

In the next step, we consider {\sc Monotone 3-Sat-(3,3)}. In order to show its hardness we state three lemmata below. 
The first one makes use of a construction inspired by the idea of an \textit{enforcer} for a clause described by Berman et al.~\cite[p.\,3]{Berman2003}. Note that we only use monotone clauses. This will require us to define a second enforcer in order to prevent the introduction of mixed clauses. 

\begin{lem}\label{lem:S(x,y,z)} Let $\mathcal{S}(x, y, z)$ be defined as the set containing the following clauses, where $V_\text{aux} = \{a, b, \ldots, f\}$ are new variables. 

\begin{multicols}{3} 
\begin{enumerate}
\item $\{x, a, b\}$
\item $\{y, c, d\}$ 
\item $\{z, e, f\}$
\item $\{a, c, f\}$
\item $\{a, d, e\}$
\item $\{b, c, e\}$
\item $\{b, d, f\}$
\item $\{\neg a, \neg c, \neg f\}$
\item $\{\neg a, \neg d, \neg e\}$
\item $\{\neg a, \neg e, \neg f\}$
\item $\{\neg b, \neg c, \neg d\}$
\item $\{\neg b, \neg c, \neg e\}$
\item $\{\neg b, \neg d, \neg f\}$
\end{enumerate}
\end{multicols}
 
Then, a truth assignment $\beta$ for $\{x, y, z\}$ can be extended to a truth assignment $\beta'$ for $\{x, y, z\} \cup V_{\text{aux}}$ that satisfies $\mathcal{S}(x, y, z)$ if and only if $\beta(v) = T$ for at least one $v \in \{x, y, z\}$.
\end{lem}

\begin{proof}
First, consider the truth assignment $\beta$ for $\{x, y, z\}$ with $\beta(x) = \beta(y) = \beta(z) = F$ and assume towards a contradiction that $\beta$ can be extended to a truth assignment $\beta'$ for $\{x, y, z\} \cup V_{\text{aux}}$ that satisfies $\mathcal{S}(x, y, z)$. Then, by clauses 1, 2 and 3, there are three variables  

\[
(u, v, w) \in \{a, b\} \times \{c, d\} \times \{e, f\}
\]
with $\beta'(u) = \beta'(v) = \beta'(w) = T$. By clauses 4 and 8, $\beta'$ nae-satisfies $\{a, c, f\}$. Analogously, $\beta'$ nae-satisfies $\{a, d, e\}, \{b, c, e\}$ and $\{b, d, f\}$, respectively. Hence, 
\[
(u, v, w) \not \in \{(a, c, f),\, (a, d, e),\, (b, c, e),\, (b, d, f)\}. 
\]
Let us consider the four remaining cases.
\begin{itemize}
\item Case $(u, v, w) = (a, c, e)$. Since $\beta'$ nae-satisfies $\{b, c, e\}$, $\{a, d, e\}$ and $\{a, c, f\}$, respectively, we have $\beta'(b) = \beta'(d) = \beta'(f) = F$. Thus, clause 7 is not satisfied.
\item Case $(u, v, w) = (a, d, f)$. Since $\beta'$ nae-satisfies $\{b, d, f\}$, $\{a, c, f\}$ and $\{a, d, e\}$, respectively, we have $\beta'(b) = \beta'(c) = \beta'(e) = F$. Thus, clause 6 is not satisfied.
\item Case $(u, v, w) = (b, c, f)$. Since $\beta'$ nae-satisfies $\{a, c, f\}$, $\{b, d, f\}$ and $\{b, c, e\}$, respectively, we have $\beta'(a) = \beta'(d) = \beta'(e) = F$. Thus, clause 5 is not satisfied.
\item Case $(u, v, w) = (b, d, e)$. Since $\beta'$ nae-satisfies $\{a, d, e\}$, $\{b, c, e\}$ and $\{b, d, f\}$, respectively, we have $\beta'(a) = \beta'(c) = \beta'(f) = F$. Thus, clause 4 is not satisfied.
\end{itemize}
Since each case yields a contradiction, we conclude that no extension of the truth assignment $\beta$ with $\beta(x) = \beta(y) = \beta(z) = F$ satisfies all clauses in $\mathcal{S}(x, y, z)$. Second, let $\beta$ be a truth assignment for $\{x, y, z\}$ with $\beta(v) = T$ for at least one $v \in \{x, y, z\}$. Then, depending on the truth values assigned to $x$, $y$ and $z$, at least one of the following three extensions of $\beta$ satisfies $\mathcal{S}(x, y, z)$: 
\begin{enumerate}
\item[] $\beta_x(a) = \beta_x(b) = F, \quad \beta_x(c) = \beta_x(d) = \beta_x(e) = \beta_x(f) = T$,
\item[] $\beta_y(a) = \beta_y(c) = \beta_y(d) = F, \quad \beta_y(b) = \beta_y(e) = \beta_y(f) = T$,
\item[] $\beta_z(d) = \beta_z(e) = \beta_z(f) = F, \quad \beta_z(a) = \beta_z(b) = \beta_z(c) = T$,
\end{enumerate}
where $\beta_v$, $v \in \{x, y, z\}$, satisfies $\mathcal{S}(x, y, z)$ if $\beta(v) = T$. 
\end{proof}

Lemma~\ref{lem:S(x,y,z)} straightforwardly translates into the following lemma.

\begin{lem}\label{lem:S(barx,bary,barz)}
Let $\mathcal{\bar{S}}(\bar{x}, \bar{y}, \bar{z})$ be the set of clauses over $V = \{x, y, z\} \cup V_\text{aux}$ obtained by negating every literal in $\mathcal{S}(x, y, z)$. Then, a truth assignment $\beta$ for $\{x, y, z\}$ can be extended to a truth assignment $\beta'$ for $V$ that satisfies $\mathcal{\bar{S}}(\bar{x}, \bar{y}, \bar{z})$ if and only if $\beta(v) = F$ for at least one $v \in \{x, y, z\}$.
\end{lem}

The construction of an unsatisfiable formula is done by combining two enforcers (see Berman et al.~\cite[p.\,3]{Berman2003}). 

\begin{propo}
There is an unsatisfiable instance of {\sc Monotone 3-Sat-(3,3)}.  
\end{propo}

\begin{proof}
By construction the formula defined by the clauses 
\[
\mathcal{S}(x, x, x) \cup \mathcal{\bar{S}}(\bar{x}, \bar{x}, \bar{x})
\]
is unsatisfiable. The two enforcers only share the variable $x$. Hence, each variable appears exactly three times unnegated and exactly three times negated. Also note that each clause is monotone by construction.
\end{proof}

We remark that there are smaller unsatisfiable instances of {\sc Monotone 3-Sat-(3,3)}. In particular, we show that there is an unsatisfiable instance with 9 variables.
\begin{propo}
The following instance of {\sc Monotone 3-Sat-(3,3)} with 9 variables and 18 clauses is unsatisfiable. 
\begin{multicols}{3} 
\begin{enumerate}
\item $\{\bar{a}, \bar{d}, \bar{g}\}$
\item $\{\bar{a}, \bar{f}, \bar{i}\}$
\item $\{\bar{b}, \bar{d}, \bar{h}\}$
\item $\{\bar{b}, \bar{e}, \bar{f}\}$
\item $\{\bar{c}, \bar{e}, \bar{g}\}$
\item $\{\bar{c}, \bar{h}, \bar{i}\}$
\item $\{a, d, g\}$
\item $\{a, f, i\}$
\item $\{b, d, h\}$
\item $\{b, e, f\}$
\item $\{c, e, g\}$
\item $\{c, h, i\}$
\item $\{a, b, c\}$
\item $\{d, e, i\}$
\item $\{f, g, h\}$
\item $\{\bar{a}, \bar{e}, \bar{h}\}$
\item $\{\bar{b}, \bar{g}, \bar{i}\}$
\item $\{\bar{c}, \bar{d}, \bar{f}\}$
\end{enumerate}
\end{multicols}
\end{propo}
\begin{proof}
Let $C$ denote the set of clauses of the instance defined above. We use clause 13 for a case analysis. Assume towards a contradiction that there is a satisfying truth assignment $\beta$ for the instance above. Clearly, for at least one variable $v \in \{a, b, c\}$ we have $\beta(v) = T$. Let us consider the remaining cases:

\noindent
\underline{$\beta(a) = T$, $\beta(b) = F$, $\beta(c) = F$:} Then, by removing satisfied clauses and unsatisfied literals, the instance reduces to the following clauses:
\begin{multicols}{3}
\begin{enumerate}[label=(\roman*)]
\item $\{\bar{d}, \bar{g}\}$
\item $\{\bar{f}, \bar{i}\}$
\item $\{\bar{e}, \bar{h}\}$
\item $\{d, h\}$
\item $\{e, f\}$
\item $\{e, g\}$
\item $\{h, i\}$
\item $\{d, e, i\}$
\item $\{f, g, h\}$
\end{enumerate}
\end{multicols}
Now, by clauses (i) to (vii), we get the following cyclic implication chain: 
\[
d \overset{\text{(i)}}{\Rightarrow} \bar{g} \overset{\text{(vi)}}{\Rightarrow} e \overset{\text{(iii)}}{\Rightarrow} \bar{h} \overset{\text{(vii)}}{\Rightarrow} i \overset{\text{(ii)}}{\Rightarrow} \bar{f} \overset{\text{(v)}}{\Rightarrow} e \overset{\text{(iii)}}{\Rightarrow} \bar{h} \overset{\text{(iv)}}{\Rightarrow} d. 
\]
Thus, $\beta(d) = \beta(e) = \beta(i) = \beta(\bar{f}) = \beta(\bar{g}) = \beta(\bar{h})$. Consequently, either clause (viii) or clause (ix) is not satisfied which is a contradiction to the assumption that $\beta$ is a satisfying truth assignment. Most of the other cases can be shown similarly (e.g., in some cases a smaller cyclic implication chain is used to infer an additional 2-clause such that a larger cyclic implication chain can be formed that yields a contradiction with the clauses $\{d, e, i\}$ and $\{f, g, h\}$). The one exception to this approach is the last case, which we consider next. 

\noindent
\underline{$\beta(a) = T$, $\beta(b) = T$, $\beta(c) = T$:} Then, the instance reduces to 
\begin{multicols}{4}
\begin{enumerate}[label=(\roman*)]
\item $\{\bar{d}, \bar{g}\}$
\item $\{\bar{f}, \bar{i}\}$
\item $\{\bar{d}, \bar{h}\}$
\item $\{\bar{e}, \bar{f}\}$
\item $\{\bar{e}, \bar{g}\}$
\item $\{\bar{h}, \bar{i}\}$
\item $\{\bar{e}, \bar{h}\}$
\item $\{\bar{g}, \bar{i}\}$
\item $\{\bar{d}, \bar{f}\}$
\item $\{d, e, i\}$
\item $\{f, g, h\}$
\end{enumerate}
\end{multicols}
Observe that the set containing the clauses (i) to (ix) is equal to $\{\bar{d}, \bar{e}, \bar{i}\} \times \{\bar{f}, \bar{g}, \bar{h}\}$. Hence, by setting one variable in clause (x) true, we have to set all variables in (xi) false. Consequently, $\beta$ can not simultaneously satisfy clauses (x) and (xi), a contradiction. We conclude that $C$ is unsatisfiable.
\end{proof}

The third and last lemma for our hardness proof of {\sc Monotone 3-Sat-(3,3)} is stated as follows. 

\begin{lem}\label{lem:setA} Let $\mathcal{A}(\bar{x}, \bar{y})$ be defined as the set containing the clauses below, where $V_\text{aux} = \{a, b, c, d\}$ are new variables. 

\begin{multicols}{4} 
\begin{enumerate}
\item $\{\bar{a}, \bar{b},\bar{x}\}$
\item $\{\bar{a}, \bar{c},\bar{x}\}$
\item $\{\bar{a}, \bar{d},\bar{x}\}$
\item $\{\bar{b}, \bar{c},\bar{y}\}$
\item $\{\bar{b}, \bar{d},\bar{y}\}$
\item $\{\bar{c}, \bar{d},\bar{y}\}$
\item $\{a, b, c\}$
\item $\{a, b, d\}$
\item $\{a, c, d\}$
\item $\{b, c, d\}$
\end{enumerate}
\end{multicols}
Then, a truth assignment $\beta$ for $\{x, y\}$ can be extended to a truth assignment $\beta'$ for $\{x, y\} \cup V_{\text{aux}}$ that satisfies $\mathcal{A}(\bar{x}, \bar{y})$ if and only if $\beta(v) = F$ for at least one $v \in \{x, y\}$.  
\end{lem}

\begin{proof}
There are $\binom{4}{2} = 6$ distinct negative 2-clauses over $V_\text{aux} = \{a, b, c, d\}$:
\[
C_2 = \{\{\bar{a}, \bar{b}\}, \{\bar{a}, \bar{c}\}, \{\bar{a}, \bar{d}\}, \{\bar{b}, \bar{c}\}, \{\bar{b}, \bar{d}\}, \{\bar{c}, \bar{d}\}\}.
\]
Consequently, if at least two variables in $V_\text{aux}$ are set true, then $C_2$ is not satisfied. Now, clauses 7, 8, 9 and 10 are satisfied if and only if we set at least two variables true (any pair of distinct variables works). Hence, if $x$ and $y$ are both set true, then $\mathcal{A}(\bar{x}, \bar{y})$ is unsatisfiable since clauses $1, 2, \ldots, 6$ are equivalent to $C_2$ in this case. Moreover, by setting $x$ or $y$ false, some of the first six clauses are satisfied (at least three of them). Now, we can choose any of these clauses, say $c_j$, and set the two variables in $c_j \cap V_\text{aux}$ true and the variables in $V_\text{aux} \setminus c_j$ false, respectively. It is easy to see that this assignment satisfies all clauses in~$\mathcal{A}(\bar{x}, \bar{y})$.    
\end{proof}

\begin{thm}\label{thm:Mon3Sat-(3,3)}
{\sc Monotone 3-Sat-(3,3)} is NP-complete.
\end{thm}

\begin{proof}
We show NP-hardness by reduction from \textsc{3-Sat-(2,2)}, for which NP-hardness was established by Berman et al.~\cite[Theorem 1]{Berman2003}. Let $\mathcal{I} = (V, C)$ be an instance of \textsc{3-Sat-(2,2)}. Let $n := |V|$ denote the number of variables, $m := |C|$ the number of clauses and let the set of variables be given as $V := \{x_1,x_2,\ldots,x_n\}$. For each variable $x_i \in V$, we introduce two new variables $x_{i,1}$, $x_{i,2}$ and replace the two negated appearances with $x_{i, 1}$ and the two unnegated appearances with $x_{i, 2}$, respectively. Then, we remove all negations and introduce
\[
\bigcup_{i=1}^n \left( \{\{x_{i,1},\,x_{i,2}\}\} \cup \mathcal{A}(\overline{x_{i,1}}, \overline{x_{i,2}}) \right ).
\]
Note that each variable appears exactly three times unnegated and three times negated, respectively (see also Lemma~\ref{lem:setA}). We introduced $n$ positive 2-clauses. Now, since $4n = 4|V| = 3|C| = 3m$, the number of these clauses is a multiple of 3. With $\mathcal{S}(\bar{y}, \bar{y}, \bar{y})$ we get 3 copies of some new variable $y$ that all have the forced truth value false (see Lemma~\ref{lem:S(barx,bary,barz)}). Then, we replace the first three 2-clauses, say $c, c', c''$, with $c \cup y$, $c' \cup y$ and $c'' \cup y$, respectively. We repeat this step until no 2-clause is left, which is possible since the number of 2-clauses is a multiple of 3. The resulting formula is an instance of {\sc Monotone 3-Sat-(3,3)} and satisfiable if and only if the original formula is satisfiable.    
\end{proof}

As shown below, Theorem~\ref{thm:Mon3Sat-(3,3)} can be extended to show that, in fact for any choice of $k\geq 3$,  \textsc{Monotone 3-Sat-($k$,$k$)} is NP-complete.

\begin{lem}
If \textsc{Monotone 3-Sat-($k$,$k$)} is NP-hard for a fixed positive integer $k$, then so is \textsc{Monotone 3-Sat-($k+1$,$k+1$)}. 
\end{lem}

\begin{proof}
We present a polynomial reduction from \textsc{Monotone 3-Sat-($k$,$k$)} to \textsc{Monotone 3-Sat-($k+1$,$k+1$)}. Given an instance of \textsc{Monotone 3-Sat-($k$,$k$)} with a set of clauses $C = \{c_1, \ldots, c_m\}$ over variables $V = \{x_1, \ldots, x_n\}$, we make $k + 1$ copies of $C$ and $V$, respectively, such that each copy has only new variables that are not shared with other copies: For $i \in \{1, \ldots, k+1\}$ let
\[
C_i = \{c_1^i, \ldots, c_m^i\}, \quad V_i = \{x_1^i, \ldots, x_n^i\}.
\]  
Next, we introduce the clauses 
\[
C_\text{inc} = \bigcup_{i = 1}^{k+1}\bigcup_{j = 1}^n \{\{x_j^i, y_j, z_j\}, \{\overline{x_j^i}, \overline{y_j}, \overline{z_j}\}\}, 
\]
where $y_j, z_j$ with $j \in \{1, \ldots, n\}$ are new variables. Now, consider the instance
\[
C' = C_\text{inc} \cup \bigcup_{i=1}^{k+1}C_i, \quad V' = \{y_j, z_j \mid 1 \leq j \leq n\} \cup \bigcup_{i=1}^{k+1}V_i.   
\]
Observe that each variable in $V'$ appears exactly $k+1$ times negated and $k+1$ times unnegated in $C'$. Further, each clause in $C'$ contains exactly three distinct literals, either all of them or none of them negated. Hence, we constructed an instance of \textsc{Monotone 3-Sat-($k+1$,$k+1$)}. Now, by setting all $y_j$ true and all $z_j$ false, respectively, we satisfy all clauses in $C_\text{inc}$. Thus, $C'$ is satisfiable if and only if $\bigcup_{i=1}^{k+1}C_i$ is satisfiable. Since the latter set of clauses is a union of disjoint copies of $C$, we conclude: $C'$ is satisfiable if and only if $C$ is satisfiable. Finally, we have
\[
|C'| = (k+1)(m + 2n), \quad |V'| = (k + 1)n + 2n = (k+3)n. 
\]
Therefore, since $k$ is a fixed positive integer, the transformation is polynomial. 

\end{proof}

\begin{cor}\label{thm:Mon3Sat-(k,k)}
{\sc Monotone 3-Sat-($k$,$k$)} is NP-complete for all $k \geq 3$. 
\end{cor}

\subsubsection{{\textsc {Monotone 3-Sat-}}$(2,2)$}

By Corollary~\ref{thm:Mon3Sat-(k,k)} we know that {\sc Monotone 3-Sat-($k$,$k$)} is NP-complete for all $k \geq 3$. Naturally, the question arises if this already settles a sharp boundary in terms of the number of variable appearances between NP-complete and polynomial time solvable cases. In this respect, this section aims at shedding light on the complexity of  \textsc{Monotone 3-Sat-(2,2)}. First of all, the question arises whether or not there are unsatisfiable instances of that problem. To the best of our knowledge, the answer to this question is still open. However, we can show that in case that question can be answered in the affirmative, \textsc{Monotone 3-Sat-(2,2)} is in fact NP-complete. We formally prove this result with the help of the following lemma.

\begin{lem}\label{lem:gadgetMCL}
Given an unsatisfiable instance of \textsc{Monotone 3-Sat-(2,2)}, we can construct a gadget $M_{C_\text{Sat}, \mathcal{L}}$ where
\begin{itemize}
\item $C_\text{Sat}$ is a set of monotone 3-clauses over a set of variables $V$,
\item $\mathcal{L}$ is a multiset of the literals $L_V = \{x_i, \overline{x_i} \mid x_i \in V\}$, 
\end{itemize}
such that the following three conditions are met: 
\begin{enumerate}
\item[(M1)] $C_\text{Sat}$ is satisfiable. Moreover, each truth assignment $\beta\colon V \rightarrow \{T, F\}$ that satisfies $C_\text{Sat}$ does not satisfy any of the literals contained in $\mathcal{L}$.
\item[(M2)] Let $\mathcal{L} = \mathcal{L}_+ \cup \mathcal{L}_-$ be the partition of $\mathcal{L}$ where $\mathcal{L}_+$ contains the positive literals, and $\mathcal{L}_-$ contains the negative literals, respectively. Then, we have $|\mathcal{L}_+| = |\mathcal{L}_-| = 3q$ for some fixed integer $q \geq 1$. 
\item[(M3)] Let $\mathcal{L}_{C_\text{Sat}}$ denote the multiset of literals that appear in $C_\text{Sat}$. Then, for each variable $x \in V$, $\mathcal{L}_{C_\text{Sat}} \cup \mathcal{L}$ contains $x$ exactly twice as a positive literal and exactly twice as a negative literal, respectively. 
\end{enumerate}
\end{lem}

\begin{proof}
Given an unsatisfiable instance of \textsc{Monotone 3-Sat-(2,2)}, let $C'$ denote the corresponding set of clauses over variables $V' = \{x'_1, \ldots, x'_n\}$. Then, there is a strict subset $C'_\text{Sat} \subsetneq C'$ such that $C'_\text{Sat}$ is satisfiable and $C'_\text{Sat} \cup \{c\}$ is unsatisfiable for all $c \in C'\setminus C'_\text{Sat}$. Now, each variable that appears in $C' \setminus C'_\text{Sat}$ has a forced truth value, i.e., if $x'_i$ appears negated (unnegated) in $C' \setminus C'_\text{Sat}$, then any satisfying truth assignments for $C'_\text{Sat}$ sets $x'_i$ true (false). Otherwise, there a satisfying assignment for $C'_\text{Sat}$ such that a clause in $C' \setminus C'_\text{Sat}$ is satisfied which is a contradiction since, by construction, such a clause does not exist. Also observe that no variable appears both negated and unnegated in $C' \setminus C'_\text{Sat}$. Let $\mathcal{L}'_+$ denote the multiset containing the positive literals appearing in $C' \setminus C'_\text{Sat}$ and $\mathcal{L}'_-$ the multiset containing the negative literals, respectively (e.g., if a negative literal $\ell$ appears twice in $C' \setminus C'_\text{Sat}$, then $\mathcal{L}'_-$ contains two copies of $\ell$). Since all clauses contain exactly three distinct literals, the number of literals in $\mathcal{L}'_- \cup \mathcal{L}'_+$ is divisible by 3. Observe that we can only guarantee that $\mathcal{L}'_+ \neq \emptyset$ or $\mathcal{L}'_- \neq \emptyset$. Therefore, we introduce a copy of $C'$  denoted by $C''$ (where the copy of  $C'_\text{Sat}$ is denoted by $C''_\text{Sat}$) over new variables $V'' = \{x''_1, \ldots, x''_n\}$, where we negate each literal. Observe that $C''$ is an instance of \textsc{Monotone 3-Sat-(2,2)}. With $\mathcal{L}''_+$ and $\mathcal{L}''_-$ defined as above, the clauses 
\begin{equation}\label{eq:C}
C_\text{Sat} = C'_\text{Sat} \cup C''_\text{Sat}
\end{equation}
force all literals in 
\begin{equation}\label{eq:L}
\mathcal{L} = \mathcal{L}'_+ \cup \mathcal{L}''_+ \cup \mathcal{L}'_- \cup \mathcal{L}''_-
\end{equation}
to be set to false. By construction, we have 
\[
|\mathcal{L}'_+ \cup \mathcal{L}''_+| = |\mathcal{L}'_- \cup \mathcal{L}''_-| = 3q
\]
with $q \geq 1$. It is now straightforward to verify that the gadget $M_{C, \mathcal{L}}$ with $C$ and $\mathcal{L}$ as defined in Equations~\eqref{eq:C} and~\eqref{eq:L}, respectively, has properties (M1), (M2) and (M3).    
\end{proof}

\begin{thm}
\textsc{Monotone 3-Sat-(2,2)} is either trivial or NP-complete. 
\end{thm}
\begin{proof}
We sketch a polynomial reduction from \textsc{Monotone 3-Sat-(3,3)}, for which NP-hardness was established in Theorem~\ref{thm:Mon3Sat-(3,3)}, with clauses over a set of variables $V = \{x_1, x_2, \ldots, x_n\}$. For each variable $x_i \in V$, we replace each appearance with a separate new variable~$x_{i, s}$, $1 \leq s \leq 6$, such that the positive literal $x_i$ is replaced with $x_{i, 1}$, $x_{i, 3}$ and $x_{i, 5}$, respectively, and the negative literal $\overline{x_i}$ is replaced with $x_{i, 2}$, $x_{i, 4}$ and $x_{i, 6}$, respectively. We denote the resulting set of clauses by $C$. Next, for each $i \in \{1, \ldots, n\}$,  we introduce the following clauses
\[
C_i = \{\{x_{i,1}, x_{i,2}\},\, \{\overline{x_{i,2}}, \overline{x_{i,3}}\},\, \{x_{i,3}, x_{i,4}\},\, \{\overline{x_{i,4}}, \overline{x_{i,5}}\},\, \{x_{i,5}, x_{i,6}\},\, \{\overline{x_{i,6}}, \overline{x_{i,1}}\}\},
\] 
which are equivalent to the following cyclic chain of implications
\[
\overline{x_{i,1}} \Rightarrow x_{i,2} \Rightarrow \overline{x_{i,3}} \Rightarrow x_{i,4} \Rightarrow \overline{x_{i,5}} \Rightarrow x_{i,6} \Rightarrow \overline{x_{i,1}}.
\]
Hence, a truth assignment $\beta$ satisfies these clauses if and only if 
\[
\beta(x_{i,1}) = \beta(x_{i,3}) = \beta(x_{i,5}) \neq \beta(x_{i,2}) = \beta(x_{i,4}) = \beta(x_{i,6})
\]
for all $i \in \{1, 2, \ldots, n\}$. Observe that each variable $x_{i,s}$ appears exactly once unnegated and exactly once negated in $C_i$ and exactly once unnegated in the remaining clauses. In order to increase the number of negated appearances of each variable by one, we introduce
\[
C'_i = \{\{\overline{x_{i,1}}, \overline{x_{i,2}}, \overline{x_{i,6}}\}, \{\overline{x_{i,3}}, \overline{x_{i,4}}, \overline{x_{i,5}}\}\}
\]
for each $i \in \{1, 2, \ldots, n\}$. Recall that a truth assignment that satisfies $C_i$ assigns different truth values to $x_{i,s}$ and $x_{i,t}$, where $s \in \{1, 3, 5\}$ and $t \in \{2, 4, 6\}$. Hence, a truth assignment that satisfies $C_i$ also satisfies $C'_i$.

Finally, we deal with the 2-clauses introduced above. Assume that there is an unsatisfiable instance of \textsc{Monotone 3-Sat-(2,2)}. 
Then we can apply  Lemma~\ref{lem:gadgetMCL} and the gadget $M_{C_\text{Sat}, \mathcal{L}}$ used in that lemma. Recall that the corresponding set of clauses $C_\text{Sat}$ can be satisfied only by assignments that do not satisfy any of the literals contained in the multiset $\mathcal{L}$. Further, the multiset $\mathcal{L}$ contains exactly $3q$ positive and exactly $3q$ negative literals for some fixed integer $q \geq 1$. Note that if we knew that $q = 1$, then we could simply use $n$ instances of this gadget to pad all 2-clauses, i.e., $\bigcup_{i = 1}^n C_i$, since each $C_i$ contains exactly 3 positive and exactly 3 negative 2-clauses. As we can not make this assumption, we solve the parity problem as follows. First, we replace the clauses 
\[
\mathcal{C} = C \cup \bigcup_{i = 1}^n \left( C_i \cup C'_i \right)
\]
with $q$ copies $\mathcal{C}_1, \mathcal{C}_2, \ldots, \mathcal{C}_q$ such that the variables of the $k$th copy are 
\[
V_k = \{x_{i, s}^k \mid 1 \leq i \leq n \text{ and } 1\leq s \leq 6\}.
\] 
Now, the set of clauses $\bigcup_{i = 1}^q \mathcal{C}_i$ contains exactly $q \cdot 3n$ negative 2-clauses and exactly $q \cdot 3n$ positive 2-clauses. Then, we use $n$ instances of the gadget $M_{C_\text{Sat}, \mathcal{L}}$, where each instance has their own new variables, to pad these 2-clauses. To be precise, we introduce the set of clauses $\bigcup_{i = 1}^n C_\text{Sat}^i$, where $C_\text{Sat}^i$ is the set of clauses corresponding to the $i$th instance of the gadget. The corresponding multiset of literals is $\bigcup_{i=1}^n \mathcal{L}^i$ and, by Property (M2), contains exactly $n \cdot 3q$ positive literals and exactly $n \cdot 3q$ negative literals. Hence, we can pair each positive (resp. negative) 2-clause with exactly one positive (resp. negative) literal that evaluates to false by Property (M1). Note that this is a one-to-one correspondence. Finally, replace each 2-clause with this union of the 2-clause with the paired literal. By construction and Property (M3) in Lemma~\ref{lem:gadgetMCL}, the resulting instance is indeed an instance of \textsc{Monotone 3-Sat-(2,2)}. It is straightforward to verify that \textsc{Monotone 3-Sat-(2,2)} is satisfiable if and only if the given instance of \textsc{Monotone 3-Sat-(3,3)} is satisfiable.  
\end{proof}

Finally, we conclude this section by settling the computational complexity of the considered problem when the literals in the clauses do not need to be distinct, i.e., \textsc{Monotone 3-Sat*-(2,2)}. 

\begin{thm}
\textsc{Monotone 3-Sat*-(2,2)} is NP-complete. 
\end{thm}

\begin{proof} 
By reduction from \textsc{Monotone 3-Sat-(3,3)}, for which NP-hardness was established in Theorem~\ref{thm:Mon3Sat-(3,3)}. Let $\mathcal{I} = (V, C)$ be an instance of \textsc{Monotone 3-Sat-(3,3)}. Let $n := |V|$ denote the number of variables, $m := |C|$ the number of clauses and let the set of variables be given as $V := \{x_1,x_2,\ldots,x_n\}$. 

For each variable $x_i \in V$, we introduce new variables $x_{i,j}$, $1 \leq j \leq 6$ and replace the three unnegated appearances with $x_{i,1}$, $x_{i,2}$ and $x_{i,3}$, and the  negated appearances, i.e., $\overline{x_i}$ with $\overline{x_{i,2}}$, $\overline{x_{i,4}}$ and $\overline{x_{i,6}}$, respectively. Next, we introduce for each $i \in \{1, 2, \ldots, n\}$ the following clauses
\begin{align*} &\{x_{i,1}, y_{i,9}, y_{i,9}\},\{\overline{x_{i,1}}, \overline{y_{i,1}}, \overline{y_{i,1}}\}, \{\overline{x_{i,1}}, \overline{y_{i,2}}, \overline{y_{i,2}}\},\\ &\{x_{i,2}, y_{i,1}, y_{i,1}\},\{x_{i,2}, y_{i,2}, y_{i,2}\}, \{\overline{x_{i,2}}, \overline{y_{i,3}}, \overline{y_{i,3}}\},\\ &\{x_{i,3}, y_{i,3}, y_{i,3}\},\{\overline{x_{i,3}}, \overline{y_{i,4}}, \overline{y_{i,4}}\}, \{\overline{x_{i,3}}, \overline{y_{i,5}}, \overline{y_{i,5}}\},\\ &\{x_{i,4}, y_{i,4}, y_{i,4}\},\{x_{i,4}, y_{i,5}, y_{i,5}\}, \{\overline{x_{i,4}}, \overline{y_{i,6}}, \overline{y_{i,6}}\},\\
&\{x_{i,5}, y_{i,6}, y_{i,6}\},\{\overline{x_{i,5}}, \overline{y_{i,7}}, \overline{y_{i,7}}\}, \{\overline{x_{i,5}}, \overline{y_{i,8}}, \overline{y_{i,8}}\},\\
&\{x_{i,6}, y_{i,7}, y_{i,7}\},\{x_{i,6}, y_{i,8}, y_{i,8}\}, \{\overline{x_{i,6}}, \overline{y_{i,9}}, \overline{y_{i,9}}\}.
\end{align*} 

Note that each variable appears exactly twice unnegated and exactly twice negated (duplicates in clauses are counted as separate appearances) in the constructed instance, and that each clause is monotone and contains exactly three literals. By construction, a truth assignment $\beta$ for $\{x_{i,s} \mid 1 \leq s \leq 6\}$ can be extended to a truth assignment $\beta'$ for $\{x_{i,s} \mid 1 \leq s \leq 6\} \cup \{y_{i,t} \mid 1 \leq t \leq 9\}$ that satisfies the clauses defined above if and only if $\beta(x_{i,1}) = \beta(x_{i,2}) = \ldots = \beta(x_{i,6})$. To that end, observe that a subset of the clauses introduced above is equivalent to the following cyclic chain of implications (see also Figure~\ref{fig:duplicates}):
\begingroup
\renewcommand{\arraystretch}{1.25}
\[
\begin{array}[c]{ccccccccccc}
x_{i,1} & \Rightarrow & \overline{y_{i,1}} & \Rightarrow & x_{i,2} & \Rightarrow & \overline{y_{i,3}} & \Rightarrow & x_{i,3} & \Rightarrow & \overline{y_{i,4}}\\
\Uparrow & & & & & & & & & & \Downarrow \\
\overline{y_{i,9}} & \Leftarrow & x_{i,6} & \Leftarrow & \overline{y_{i,7}} & \Leftarrow & x_{i,5} & \Leftarrow & \overline{y_{i,6}} & \Leftarrow & x_{i,4}
\end{array}
\]
\endgroup
Hence, we have $\beta(x_{i,1}) = \beta(x_{i,2}) = \ldots = \beta(x_{i,6})$ in any satisfying truth assignment. Furthermore, we can satisfy the introduced clauses by setting all variables in $\{x_{i,s} \mid 1 \leq s \leq 6\}$ true (resp. false) and all variables in $\{y_{i,t} \mid 1 \leq t \leq 9\}$ false (resp. true).

Now it is straightforward to verify that we constructed an instance of \textsc{Monotone 3-Sat*-(2,2)} that is satisfiable if and only if $\mathcal{I}$ is satisfiable.
\begin{figure}    
  \centering
  \includegraphics[width=0.7\textwidth]{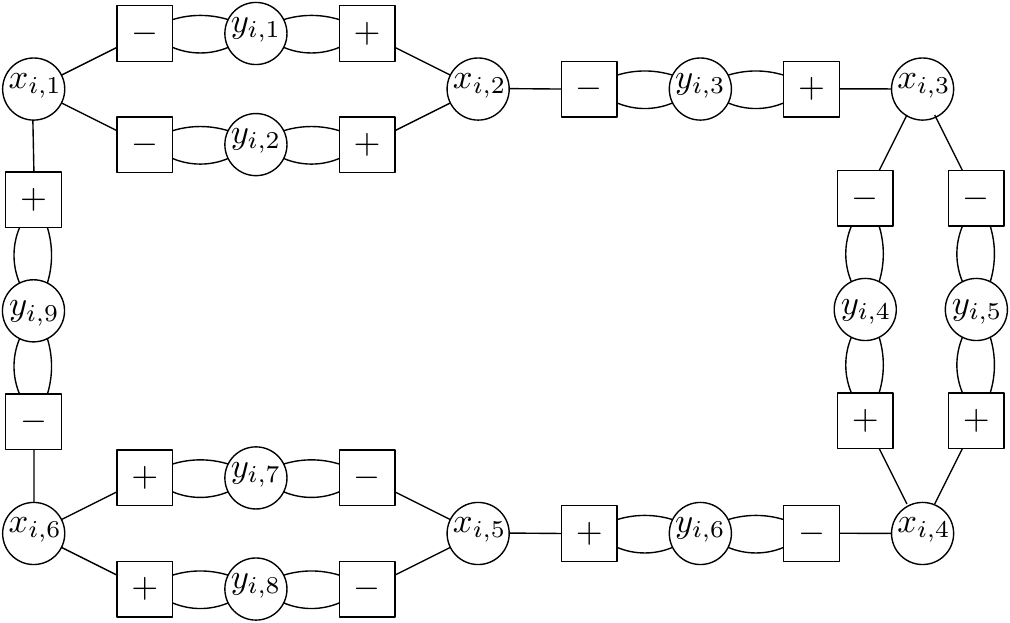}
  \caption{Gadget to reduce the number of variable appearances of a variable $x_i$. The variables $x_{i,1}, \ldots, x_{i,6}$ correspond to the six appearances of $x_i$ and $y_{i,1}, \ldots, y_{i,9}$ are auxiliary variables. Variables are depicted as circles and clauses as squares, respectively. All clauses are monotone, i.e., either all variables in a clause are unnegated (indicated by the plus sign) or all of them are negated (indicated by the minus sign), respectively. An edge between a variable and a clause means that the variable appears in the clause, either once if there is exactly one edge or twice if there are two edges. By construction, a truth assignment $\beta$ satisfies the depicted clauses if and only if $\beta(x_{i,1}) = \beta(x_{i,2}) = \ldots = \beta(x_{i,6}) \neq \beta(y_{i,1}) = \beta(y_{i,2}) = \ldots = \beta(y_{i,9})$.}
  \label{fig:duplicates}
\end{figure}
\end{proof}

\subsection{Exactly one negated appearance of each variable}\label{sub:once-negated}

In this section, we settle the computational complexity status of \textsc{Monotone 3-Sat-($k,1$)} for $k\geq 5$. We do not answer the question of its computational complexity for $k \in \{3,4\}$, which, to the best of our knowledge, is still open. However, for $k \in \{3,4\}$ we can show that when restricted to a ``small'' number of unnegated appearances each instance of  \textsc{Monotone 3-Sat-($k,1$)} is satisfiable. 

\subsubsection{On {\textsc{Monotone 3-Sat-($k,1$)}} for $k \geq 5$}

It will be useful to introduce some additional notation. Let $V$ be a set of variables and $C, C' \subseteq \mathcal{P}(V)$ non-empty sets of clauses, where $\mathcal{P}(V)$ denotes the power set of $V$. We say that $C$ \emph{subsumes} $C'$ if for each clause $c' \in C'$ there is a clause $c \in C$ such that $c \subseteq c'$. Consequently, if $C$ is satisfiable and subsumes $C'$, then $C'$ is satisfiable. On the other hand, if $C'$ is unsatisfiable, then so is $C$. 

We begin with Lemma~\ref{lem:D(X)} which will be used for proving the computational complexity result for \textsc{Monotone 3-Sat-($k,1$)} for $k\geq 5$.

\begin{lem}\label{lem:D(X)}
Let $\mathcal{D}(X)$ with $X = (x_1, x_2, \ldots, x_6)$ be the following set of clauses, where $V_\text{aux} = \{a, b, \ldots, i\}$ are new variables. 

\begin{multicols}{4} 
\begin{enumerate}
\item $\{\bar{a}, \bar{c}, \bar{e}\}$
\item $\{\bar{b}, \bar{f}, \bar{h}\}$
\item $\{\bar{d}, \bar{g}, \bar{i}\}$
\item $\{a, b, d\}$
\item $\{a, d, f\}$
\item $\{a, f, i\}$
\item $\{a, h, i\}$
\item $\{b, c, d\}$
\item $\{b, c, g\}$
\item $\{b, e, g\}$
\item $\{c, g, h\}$
\item $\{c, h, i\}$
\item $\{e, f, g\}$
\item $\{e, f, i\}$
\item $\{a, g, x_1\}$
\item $\{b, i, x_2\}$
\item $\{c, f, x_3\}$
\item $\{d, e, x_4\}$
\item $\{d, h, x_5\}$
\item $\{e, h, x_6\}$
\end{enumerate}
\end{multicols}
Then, a truth assignment $\beta$ for $X$ can be extended to a truth assignment $\beta'$ for $X \cup V_{\text{aux}}$ that satisfies $\mathcal{D}(X)$ if and only if $\beta(v) = T$ for at least one $v \in X$.
\end{lem}
\begin{proof}
Since each new variable appears only once negated, we can assume that a truth assignment that satisfies $\mathcal{D}(X)$ assigns the truth value false to exactly one variable of each negative clause (the corresponding literal evaluates to true). Hence, clauses 1, 2 and 3 in conjunction with the following set of $3^3 = 27$ clauses 
\[
\mathcal{U} = \{\{u, v, w\} \mid (u,v,w) \in \{a, c, e\} \times \{b,f,h\} \times \{d,g,i\}  \}
\]
is unsatisfiable.  
For now, we consider $\mathcal{D}(X)$ with the variables $x_1, \ldots, x_6$ removed (i.e., clauses 15--20 are 2-clauses). Let $\mathcal{D}(\emptyset)$ denote this set of clauses. Next, we show that $\mathcal{D}(\emptyset)$ subsumes $\mathcal{U}$. For each clause $c \in \mathcal{U}$ we list the clauses in $\mathcal{D}(\emptyset)$ that subsume $c$ (e.g., the clause $\{e, h, d\}$ is subsumed by clause 18, 19 and 20, respectively): 
\begin{multicols}{3} 
\begin{enumerate}
\item $\{a, b, d\}$ (4)
\item $\{a, b, g\}$ (15)
\item $\{a, b, i\}$ (16)
\item $\{c, b, d\}$ (8)
\item $\{c, b, g\}$ (9)
\item $\{c, b, i\}$ (16)
\item $\{e, b, d\}$ (18)
\item $\{e, b, g\}$ (10)
\item $\{e, b, i\}$ (16)
\item $\{a, f, d\}$ (5)
\item $\{a, f, g\}$ (15)
\item $\{a, f, i\}$ (6)
\item $\{c, f, d\}$ (17)
\item $\{c, f, g\}$ (17)
\item $\{c, f, i\}$ (17)
\item $\{e, f, d\}$ (18)
\item $\{e, f, g\}$ (13)
\item $\{e, f, i\}$ (14)
\item $\{a, h, d\}$ (19)
\item $\{a, h, g\}$ (15)
\item $\{a, h, i\}$ (7)
\item $\{c, h, d\}$ (19)
\item $\{c, h, g\}$ (11)
\item $\{c, h, i\}$ (12)
\item $\{e, h, d\}$ (18,19,20)
\item $\{e, h, g\}$ (20)
\item $\{e, h, i\}$ (20)
\end{enumerate}
\end{multicols}
Thus, $\mathcal{D}(\emptyset)$ is unsatisfiable. Consequently, if all $x_s$, $1 \leq s \leq 6$, are set false, then $\mathcal{D}(X)$ is unsatisfiable. Now, if we satisfy any of the clauses 15--20 in $\mathcal{D}(X)$ by setting $x_s$ to true for some $s \in \{1,2,\ldots,6\}$, then at least one of the clauses in $\mathcal{U}$ is not subsumed. Note that removing any clause of $\mathcal{D}(\emptyset)\setminus\{c_1, c_2, c_3\}$, where $c_j$ refers to the $j$th clause in $\mathcal{D}(X)$, means that some clause in $\mathcal{U}$ is not subsumed. In other words, if we remove any positive clause, we can satisfy the remaining clauses in $\mathcal{D}(X)$. 
Let $\{u,v,w\} \in \mathcal{U}$ be a clause that is not subsumed after assigning truth values to $x_1, \ldots x_6$. Then, setting $u, v, w$ false and all other variables true satisfies all clauses in $\mathcal{D}(X)$. Note that this assignment satisfies all negative clauses by construction. Further, each positive clause contains at least one variable $z \not\in \{u,v,w\}$, and thus is satisfied.  
\end{proof}

\textbf{Remark.} Let  $y$ be a new variable. By construction, the set of clauses   
\[
\mathcal{F}(y) = \mathcal{D}(X_1) \cup \mathcal{D}(X_2) \cup \mathcal{D}(X_3) \cup \{\{\overline{u_1}, \overline{u_2}, \overline{u_3}\}\},
\]
with new variables $u_1, u_2, u_3$ and $X_i = (y, u_i, u_i, u_i, u_i, u_i)$ for $i \in \{1, 2, 3\}$ forces $y$ to true, where  $\mathcal{D}(X_i)$ refers to the set of clauses in Lemma~\ref{lem:D(X)}. Note that each new variable except $y$ appears exactly 5 times unnegated and once negated and $y$ appears exactly three times unnegated. 

Now, we are ready to prove hardness of \textsc{Monotone 3-Sat-(5,1)}.

\begin{thm}\label{thm:Mon3Sat-(1,5)}
\textsc{Monotone 3-Sat-(5,1)} is NP-complete. 
\end{thm}
\begin{proof}
By reduction from \textsc{3-Sat-(2,2)}, for which NP-hardness was established by Berman et al.~\cite[Theorem 1]{Berman2003}. Given an instance of the latter with a set $V$ of variables and a set $C$ of clauses over $V$, let $n:=|V|$. For each variable $x_i \in V$, we introduce two new variables $x_{i,1}$, $x_{i,2}$ and replace the two negated appearances with $x_{i, 1}$ and the two unnegated appearances with $x_{i, 2}$, respectively. Then, we remove all negations and introduce the following clauses for $i \in \{1,2, \ldots, n\}$:
\[
\mathcal{D}(x_{i, 1}, x_{i, 1}, x_{i, 1}, x_{i, 2}, x_{i, 2}, x_{i, 2}) \cup \{\{\overline{x_{i,1}}, \overline{x_{i,2}} , \overline{y_i}\}\} \cup \mathcal{F}(y_i),
\]   
where $y_i$ is a new variable. Since these clauses can be satisfied if and only if we assign different truth values to $x_{i,1}$ and $x_{i,2}$, the resulting formula is a yes-instance if and only if the original formula is a yes-instance. By construction, all variables except $y_i$, $1 \leq i \leq n$, appear exactly 5 times unnegated and once negated. Recall that $4n =  3|C|$ holds in the given instance of \textsc{3-Sat-(2,2)}, and hence the number of variables $n$ is a multiple of 3. Thus, we have $n = 3q$ variables $y_i$ that each appears exactly three times unnegated and exactly once negated. We increase the number of unnegated appearances to the desired number 5 by introducing the clauses (for $q > 1$)
\[
\bigcup_{i = 1}^{q}\{\{y_{3i-2}, y_{3i-1}, y_{3i}\}\} \cup \bigcup_{i = 1}^{q-1} \{\{y_{3i-1}, y_{3i}, y_{3i + 1}\}\} \cup \{y_{n-1}, y_n, y_1\}. 
\]
Since the variable $y_i$ have the forced truth value true, these clauses have no effect on the constructed formula. Moreover, for $q > 1$ these clauses are pairwise distinct (i.e., there are no identical clauses). For $q = 1$, we can use
\[
\mathcal{D}((y_1, y_1, y_2, y_2, y_3, y_3)).
\]
Note that we can handle the case $q > 1$ in the same way, but it would result in a larger construction since each instance of $\mathcal{D}(X)$ introduces 20 clauses and 9 variables. We conclude by remarking that the transformation is polynomial. 
\end{proof}

We can construct on unsatisfiable instance of \textsc{Monotone 3-Sat-(5,1)} as follows:
\[
\mathcal{F}(y_1) \cup \mathcal{F}(y_2) \cup \mathcal{F}(y_3) \cup \{\{\overline{y_1}, \overline{y_2}, \overline{y_3}\}\} \cup \mathcal{D}(y_1, y_1, y_2, y_2, y_3, y_3).
\]
Hence, we get the following proposition. 

\begin{propo}
There exists an unsatisfiable instance of \textsc{Monotone 3-Sat-(5,1)} with 204 clauses and 102 variables. 
\end{propo}

Next, with the lemma below we show that Theorem~\ref{thm:Mon3Sat-(1,5)} implies hardness of \textsc{Monotone 3-Sat-($k,1$)}, for any choice of $k\geq 5$.

\begin{lem}
If \textsc{Monotone 3-Sat-($k,1$)} is NP-hard for a fixed positive integer $k$, then so is \textsc{Monotone 3-Sat-($k+1,1$)}. 
\end{lem}

\begin{proof}
We present a polynomial reduction from \textsc{Monotone 3-Sat-($k,1$)} to \textsc{Monotone 3-Sat-($k+1,1$)}. Given an instance of \textsc{Monotone 3-Sat-($k,1$)} with a set of clauses $C = \{c_1, \ldots, c_m\}$ over variables $V = \{x_1, \ldots, x_n\}$, we make $k + 1$ copies of $C$ and $V$, respectively, such that each copy has only new variables that are not shared with other copies: For $i \in \{1, \ldots, k+1\}$ let
\[
C_i = \{c_1^i, \ldots, c_m^i\}, \quad V_i = \{x_1^i, \ldots, x_n^i\}.
\]  
Note that the number of variables $n$ is divisible by 3 since each variable appears exactly once negated and all negative 3-clauses contain exactly three negated variables. Hence, $n = 3q$ for some positive integer $q$. We introduce the clauses 
\[
C_\text{inc} = \bigcup_{i = 1}^{k+1}\bigcup_{j = 1}^n \{\{x_j^i, y_j, z_j\}\} \cup \bigcup_{i = 1}^q \{\{\overline{y_{3i-2}}, \overline{y_{3i-1}}, \overline{y_{3i}}\}, \{\overline{z_{3i-2}}, \overline{z_{3i-1}}, \overline{z_{3i}}\}\}, 
\]
where $y_j, z_j$ with $j \in \{1, \ldots, n\}$ are new variables. Now, consider the instance
\[
C' = C_\text{inc} \cup \bigcup_{i=1}^{k+1}C_i, \quad V' = \{y_j, z_j \mid 1 \leq j \leq n\} \cup \bigcup_{i=1}^{k+1}V_i.   
\]
Observe that each variable in $V'$ appears exactly once negated and $k+1$ times unnegated in $C'$. Further, each clause in $C'$ contains exactly three distinct literals, either all of them or none of them negated. Hence, we constructed an instance of \textsc{Monotone 3-Sat-($k+1,1$)}. Now, by setting $y_j$ true and $z_j$ false if $j$ is even; and $y_j$ false and $z_j$ true if $j$ is odd, respectively, we satisfy all clauses in $C_\text{inc}$. Thus, $C'$ is satisfiable if and only if $\bigcup_{i=1}^{k+1}C_i$ is satisfiable. Since the latter set of clauses is a union of disjoint copies of $C$, we conclude: $C'$ is satisfiable if and only if $C$ is satisfiable. Finally, we have
\[
|C'| = (k+1)(m + n) + 2q, \quad |V'| = (k + 1)n + 2n = (k+3)n. 
\]
Therefore, since $k$ is a fixed positive integer, the transformation is polynomial.  
\end{proof}

\begin{cor}
\textsc{Monotone 3-Sat-($k,1$)} is NP-complete for all $k \geq 5$. 
\end{cor}

\subsubsection{On \textsc{Monotone 3-Sat-($3,1$)} and \textsc{Monotone 3-Sat-($4,1$)}}

We now discuss some properties of \textsc{Monotone 3-Sat-($k,1$)}, and conclude the section with corollaries stating that for certain ``small'' numbers of variable appearances each instance of \textsc{Monotone 3-Sat-($3,1$)} and \textsc{Monotone 3-Sat-($4,1$)} is satisfiable.

Let $V = \{x_1, x_2, \ldots, x_n\}$ be a set of variables.  First, the number of variables $n = |V|$ is divisible by 3 since otherwise there is a negative clause containing less than three variables. Next, we expand on ideas presented in Lemma~\ref{lem:D(X)}. Recall that we can restrict our attention to truth assignments that set exactly one literal in each negative clause to true (i.e., the corresponding variable to false), since we can simply modify any satisfying truth assignment to meet that requirement. Moreover, we can assume that (after relabeling) the negative clauses are
\[
\{\overline{x_1}, \overline{x_2}, \overline{x_3}\}, \{\overline{x_4}, \overline{x_5}, \overline{x_6}\}  \ldots, \{\overline{x_{n-2}}, \overline{x_{n-1}}, \overline{x_n}\}. 
\] 
Hence, we can represent any truth assignment of interest by a tuple  
\[
(x_{i_1}, \ldots, x_{i_{\frac{n}{3}}}) \in \{x_1, x_2, x_3\} \times \{x_4, x_5, x_6\} \times \ldots \times \{x_{n-2}, x_{n-1}, x_n\},
\]
such that the corresponding truth assignment $\beta \colon V \rightarrow \{T, F\}$ is defined as $\beta(x_j) = F$ if and only if $j \in \{i_1, i_2,\ldots, i_{\frac{n}{3}}\}$. It is convenient to define 
\[
\mathcal{M}_n = \{\{x_{i_1}, \ldots, x_{i_{\frac{n}{3}}}\} \mid (x_{i_1}, \ldots, x_{i_{\frac{n}{3}}}) \in \{x_1, x_2, x_3\} \times \ldots \times \{x_{n-2}, x_{n-1}, x_n\}\}
\] 
which represents the truth assignments that set exactly one literal in each negative clause to true (in an instance of \textsc{Monotone 3-Sat-($k,1$)} with $n$ variables).

Next, we define the family $\mathcal{U}_n$ of sets which is made up of 3-clauses that correspond to subsets of elements of $\mathcal{M}_n$: 
\[
\mathcal{U}_n = \bigcup_{X \in \mathcal{M}_n} \{\{S \subseteq X \mid |S| = 3\}\}
\] 
Intuitively, each element of $\mathcal{U}_n$ is the set of all 3-clauses that are not satisfied by the corresponding truth assignment (with respect to the same set of variables). Recall that every $X \in \mathcal{M}_n$ corresponds to a truth assignment.

We are now ready to state the following lemma.

\begin{lem}\label{thm:sat_nec_cond} 
Let $V = \{x_1, x_2, \ldots, x_n\}$ be a set of variables. An instance of \textsc{Monotone 3-Sat} with a collection of clauses \[
C = \{\{\overline{x_1}, \overline{x_2}, \overline{x_3}\}, \{\overline{x_4}, \overline{x_5}, \overline{x_6}\}  \ldots, \{\overline{x_{n-2}}, \overline{x_{n-1}}, \overline{x_n}\}\} \cup C^+,  
\]       
where $C^+$ is a collection of positive 3-clauses, is satisfiable if and only if there exists a $U \in \mathcal{U}_n$ such that
$
C^+ \cap U = \emptyset. 
$
\end{lem}

\begin{proof}
First, suppose there is a $U \in \mathcal{U}_n$ such that $C^+ \cap U = \emptyset$. Then, we set all variables in $X \in \mathcal{M}_n$ corresponding to $U = \{S \subseteq X \mid |S| = 3\}$ false and the other variables in $V\setminus X$ true, respectively. By construction, all negative clauses of $C$ are satisfied if we set all variables false for any $X \in \mathcal{M}_n$. Now, assume towards a contradiction that a clause $c \in C^+$ is not satisfied. Then $c \subseteq X$ with $|c| = 3$, and thus $\{c\}\subseteq C^+ \cap U$, a contradiction. Hence, all positive clauses are satisfied. Second, let $C^+ \cap U \neq \emptyset$ for all $U \in \mathcal{U}_n$. By construction, the truth assignment corresponding to $U \in \mathcal{U}_n$ does not satisfy any clause in $C^+ \cap U$. Since every satisfying truth assignment can be modified such that exactly one literal in each negative clause is set true, we conclude that no satisfying assignment for $C$ exists. 
\end{proof}

\begin{thm}
Let $V = \{x_1, x_2, \ldots, x_n\}$ be a set of variables. An instance of \textsc{Monotone 3-Sat} with a collection of clauses \[
C = \{\{\overline{x_1}, \overline{x_2}, \overline{x_3}\}, \{\overline{x_4}, \overline{x_5}, \overline{x_6}\}  \ldots, \{\overline{x_{n-2}}, \overline{x_{n-1}}, \overline{x_n}\}\} \cup C^+,  
\]       
where $C^+$ is a collection of positive 3-clauses, is satisfiable if each variable appears unnegated less than $\frac{81}{n}$ times.  
\end{thm}

\begin{proof}
Let $n = 3k$ (this is not a restriction since $n$ must be a multiple of 3 such that all negative clauses contain exactly three literals). The family $\mathcal{U}_n$ contains exactly $|\mathcal{M}_n| = 3^k$ sets of 3-clauses. Further, any 3-clause contained in some $U \in \mathcal{U}_n$ appears in $3^{k-3}$ elements of $\mathcal{U}_n$ (since three elements in the Cartesian product used in the definition of $\mathcal{M}_n$ are fixed and there are $3^{k-3}$ ways to choose the remaining elements). Now, an unsatisfying instance contains at least one clause of each $U \in \mathcal{U}_n$ (Lemma~\ref{thm:sat_nec_cond}). Since one clause covers exactly $3^{k-3}$ elements of $\mathcal{U}_n$, we need at least $\frac{3^k}{3^{k-3}} = 3^3 = 27$ clauses to cover all elements of $\mathcal{U}_n$. Note that some elements may be covered more than once (hence, we may need more 3-clauses in case this is unavoidable). As we need 27 3-clauses, we have $3 \cdot 27 = 81$ unnegated variable appearances. Thus, at least one variable appears unnegated at least $\frac{81}{n}$ times. 
\end{proof}

The above theorem allows us to derive several corollaries, most directly Corollaries~\ref{cor:Mon3Sat(1,4)} and \ref{cor:Mon3Sat(1,3)} on \textsc{Monotone 3-Sat-($4,1$)} and \textsc{Monotone 3-Sat-($3,1$)} respectively. 

\begin{cor}
Let $V = \{x_1, x_2, \ldots, x_n\}$ be a set of variables. An instance of \textsc{Monotone 3-Sat} with a collection of clauses \[
C = \{\{\overline{x_1}, \overline{x_2}, \overline{x_3}\}, \{\overline{x_4}, \overline{x_5}, \overline{x_6}\}  \ldots, \{\overline{x_{n-2}}, \overline{x_{n-1}}, \overline{x_n}\}\} \cup C^+,  
\]      
where $C^+$ is a minimum hitting set\footnote{See, e.g., Garey and Johnson~\cite[p.\,222]{garey79} for a formal definition of the hitting set problem.} for $\mathcal{U}_n$ is unsatisfiable. Here, a minimum hitting set is a set $C^+$ of positive 3-clauses of smallest size such that $C^+ \cap U \neq \emptyset$ for each $U \in \mathcal{U}_n$.
Further, every instance of \textsc{Monotone 3-Sat}, where each variable appears negated at most once, and that has at most $|C^+|-1$ clauses is satisfiable. 
\end{cor}

\begin{cor}
Let $V = \{x_1, x_2, \ldots, x_n\}$ be a set of variables. An instance of \textsc{Monotone 3-Sat} with a collection of clauses \[
C = \{\{\overline{x_1}, \overline{x_2}, \overline{x_3}\}, \{\overline{x_4}, \overline{x_5}, \overline{x_6}\}  \ldots, \{\overline{x_{n-2}}, \overline{x_{n-1}}, \overline{x_n}\}\} \cup C^+,  
\]       
where $C^+$ is a collection of positive 3-clauses and $|C^+| < 27$ is satisfiable. 
\end{cor}

\begin{cor}\label{cor:Mon3Sat(1,4)}
Each instance of \textsc{Monotone 3-Sat-($4,1$)} with less than 21 variables is satisfiable. 
\end{cor}

\begin{cor}\label{cor:Mon3Sat(1,3)}
Each instance of \textsc{Monotone 3-Sat-($3,1$)} with less than 27 variables is satisfiable. 
\end{cor}

\subsection{Dichotomy for exactly six appearances per variable}\label{sub:six}

In this section, we settle the computational complexity status for \textsc{Monotone 3-Sat-E6}, i.e., with exactly $6$ variable appearances. From the previous section we know that \textsc{Monotone 3-Sat-($5,1$)}, and hence \textsc{Monotone 3-Sat-($1,5$)} are NP-complete (Theorem~\ref{thm:Mon3Sat-(1,5)}). From Section~\ref{sub:balanced} we can conclude that \textsc{Monotone 3-Sat-(3,3)} is NP-complete (Theorem~\ref{thm:Mon3Sat-(k,k)}). Clearly, if a variable appears only unnegated or only negated the problem becomes trivial. Therefore, it remains to establish the the computational complexity status for \textsc{Monotone 3-Sat-($4,2$)} (and hence for \textsc{Monotone 3-Sat-($2,4$)}). 

In order to do so, we state two lemmata that allow, in an intermediate step, to prove hardness of \textsc{Monotone 3-Sat-($3,2$)}. That result, in turn, will be used to show hardness of \textsc{Monotone 3-Sat-($4,2$)}. 

\begin{lem} Let $\mathcal{G}(x, y, z)$ be the following set of clauses, where $V_\text{aux} = \{a, b, \ldots, f\}$ are new variables.  

\begin{multicols}{3} 
\begin{enumerate}
\item $\{\overline{a}, \overline{b}, \overline{f}\}$
\item $\{\overline{a}, \overline{c}, \overline{d}\}$
\item $\{\overline{b}, \overline{c}, \overline{e}\}$
\item $\{\overline{d}, \overline{e}, \overline{f}\}$
\item $\{a, b, f\}$
\item $\{a, c, d\}$
\item $\{b, c, e\}$
\item $\{d, e, f\}$
\item $\{a, e, x\}$
\item $\{b, d, y\}$
\item $\{c, f, z\}$
\end{enumerate}
\end{multicols}
Then, a truth assignment $\beta$ for $\{x, y, z\}$ can be extended to a truth assignment $\beta'$ for $\{x, y, z\} \cup V_{\text{aux}}$ that satisfies $\mathcal{G}(x, y, z)$ if and only if $\beta(v) = T$ for at least one $v \in \{x, y, z\}$. 
\end{lem}

\begin{proof}
First, assume towards a contradiction that there is a truth assignment $\beta' \colon \{x, y, z, a, \ldots, f\} \rightarrow \{T, F\}$ with $\beta'(x) = \beta'(y) = \beta'(z) = F$ that satisfies $\mathcal{G}(x, y, z)$. Then, there is a triple
\[
(u, v, w) \in \{a, e\} \times \{b, d\} \times \{c, f\}
\]
such that each variable in $\{u, v, w\}$ is set true (by clauses 9, 10 and 11). Now, we show that $u \neq a$. Since $(u, v, w) \not \in \{(a, b, f), (a, d, c)\}$ by clauses 1 and 2, it suffices to consider the cases $(u, v, w) \in \{(a, b, c), (a, d, f)\}$. 
\begin{itemize}
\item If $(u, v, w) = (a, b, c)$, then $\beta'(f) = \beta'(d) = \beta'(e) = F$ by clauses 1, 2 and 3, respectively. Hence, clause 8 is not satisfied which is a contradiction to the assumption that $\beta'$ satisfies $\mathcal{G}(x, y, z)$.  
\item If $(u, v, w) = (a, d, f)$, then $\beta'(b) = \beta'(c) = \beta'(e) = F$ by clauses 1, 2 and 4, respectively. Hence, clause 7 is not satisfied, which is again a contradiction to the assumption that $\beta'$ satisfies $\mathcal{G}(x, y, z)$.  
\end{itemize}
By an analogous argument, we can show that $u \neq e$ which is a contradiction since $u \in \{a, e\}$. Hence, there is no satisfying truth assignment $\beta'$ with $\beta'(x) = \beta'(y) = \beta'(z) = F$. We deduce that no extension of a truth assignment $\beta$ for $\{x, y, z\}$ with $\beta(x) = \beta(y) = \beta(z) = F$ satisfies $\mathcal{G}(x, y, z)$. Second, let $\beta$ be a truth assignment for $\{x, y, z\}$ with $\beta(x) = T$, $\beta(y) = b_y$ and $\beta(z) = b_z$ where $b_y, b_z \in \{T,F\}$. Then, we extend $\beta$ to a truth assignment $\beta'$ that satisfies $\mathcal{G}(x, y, z)$ by setting $\beta'(a) = \beta'(e) = F$ and $\beta'(v) = T$ for all $v \in V_\text{aux}\setminus\{a, e\}$. It is easy to verify that $\mathcal{G}(x, y, z)$ is satisfied for this assignment even if $b_y = b_z = F$. By using the same approach, we can show that if $\beta(y) = T$ or $\beta(z) = T$,  we can assign truth values to the remaining variables such that $\mathcal{G}(x, y, z)$ is satisfied.     

\end{proof}

\begin{lem}
Let $\mathcal{H}(\bar{x}, \bar{y}, \bar{z})$ be the following set of clauses, where $V_\text{aux} = \{a, b, \ldots, i\}$ are new variables.

\begin{multicols}{4} 
\begin{enumerate}
\item $\{\bar{a}, \bar{d}, \bar{x}\}$
\item $\{\bar{b}, \bar{g}, \bar{y}\}$
\item $\{\bar{f}, \bar{i}, \bar{z}\}$
\item $\{\bar{a}, \bar{b}, \bar{e}\}$
\item $\{\bar{c}, \bar{e}, \bar{i}\}$
\item $\{\bar{c}, \bar{g}, \bar{h}\}$
\item $\{\bar{d}, \bar{f}, \bar{h}\}$
\item $\{a, c, f\}$
\item $\{a, f, g\}$
\item $\{a, g, h\}$
\item $\{b, c, d\}$
\item $\{b, e, h\}$
\item $\{b, h, i\}$
\item $\{c, e, i\}$
\item $\{d, e, f\}$
\item $\{d, g, i\}$
\end{enumerate}
\end{multicols}
Then, a truth assignment $\beta$ for $\{x, y, z\}$ can be extended to a truth assignment $\beta'$ for $\{x, y, z\} \cup V_{\text{aux}}$ that satisfies $\mathcal{H}(\bar{x}, \bar{y}, \bar{z})$ if and only if $\beta(v) = F$ for at least one $v \in \{x, y, z\}$.
\end{lem}

\begin{proof}
First, assume towards a contradiction that there is a truth assignment $\beta' \colon \{x, y, z, a, \ldots, i\} \rightarrow \{T, F\}$ with $\beta'(x) = \beta'(y) = \beta'(z) = T$ that satisfies $\mathcal{H}(\bar{x}, \bar{y}, \bar{z})$. Then, there is a triple
\[
(u, v, w) \in \{a, d\} \times \{b, g\} \times \{f, i\}
\]
such that each variable in $\{u, v, w\}$ is set false (by clauses 1, 2 and 3). By clauses 9 and 16 we have $(u,v,w) \not \in \{(a, g, f), (d, g, i)\}$. Hence, 
\[
(u,v,w) \in \{(a, b, f), (a, b, i), (a, g, i), (d, b, f), (d, b, i), (d, g, f)\}
\]
Let us consider each of these cases. 
\begin{itemize}
\item If $(u, v, w) = (a, b, f)$, then $\beta'(c) = \beta'(g) = T$ by clauses 8 and 9, respectively. Now, by clause 6 we have $\beta'(h) = F$. By clauses 12 and 13, respectively, we have $\beta'(e) = \beta'(i) = T$. Thus, clause 5 is not satisfied which is a contradiction to the assumption that $\beta'$ satisfies $\mathcal{H}(\bar{x}, \bar{y}, \bar{z})$. 
\item If $(u, v, w) = (a, b, i)$, then $\beta'(h) = T$ by clause 13. By clause 6, we have $\beta'(c) = F$ or $\beta'(g) = F$. First, let $\beta'(c) = F$. By clause 8 and 11, respectively, we get $\beta'(f) = \beta'(d) = T$. Hence, clause 7 is not satisfied. Second, let $\beta'(g) = F$. By clauses 9 and 16, respectively, we have $\beta'(f) = \beta'(d) = T$. Again, clause 7 is not satisfied which is a contradiction to the assumption. 
\item  If $(u, v, w) = (a, g, i)$, then $\beta'(f) = \beta'(h) = \beta'(d) = T$ by clauses 9, 10 and 16, respectively. Thus, clause 7 is not satisfied, a contradiction.
\item If $(u, v, w) = (d, b, f)$, then $\beta'(c) = \beta'(e) = T$ by clause 11 and 15, respectively. Now, by clause 5 we have $\beta'(i) = F$. By clause 13 and 16, respectively, we have $\beta'(h) = \beta'(g) = T$. Thus, clause 6 is not satisfied, a contradiction. 
\item If $(u, v, w) = (d, b, i)$, then $\beta'(c) = \beta'(h) = T$ by clause 11 and 13, respectively. Now, by clause 6 we have $\beta'(g) = F$. Hence, clause 16 is not satisfied, a contradiction.
\item If $(u, v, w) = (d, g, f)$, then $\beta'(a) = \beta'(e) = \beta'(i) = T $ by clause 9, 15 and 16, respectively. By clause 4 and 5, respectively, we have $\beta'(b) = \beta'(c) = F$. Thus, clause 11 is not satisfied, a contradiction. 
\end{itemize}
Hence, there is no truth assignment $\beta'$ with $\beta'(x) = \beta'(y) = \beta'(z) = T$ that satisfies $\mathcal{H}(\bar{x}, \bar{y}, \bar{z})$. We deduce that no extension of a truth assignment $\beta$ for $\{x, y, z\}$ with $\beta(x) = \beta(y) = \beta(z) = T$ satisfies $\mathcal{H}(\bar{x}, \bar{y}, \bar{z})$. Second, let $\beta$ be a truth assignment for $\{x, y, z\}$ with $\beta(x) = b_x$, $\beta(y) = b_y$ and $\beta(z) = b_z$ where $b_x, b_y, b_z \in \{T,F\}$ and $b_v = F$ for at least one variable $v \in \{x, y, z\}$. We extend $\beta$ to a truth assignment $\beta'$ that satisfies $\mathcal{H}(\bar{x}, \bar{y}, \bar{z})$ by setting the variables in $V_\text{aux} = \{a, b, \ldots, i\}$ as follows:
\begin{itemize}
\item If $b_x = F$, we set $\beta'(v) = T$ for all $v \in \{a, c, d, e, h\}$ and $\beta'(w) = F$ for all $w \in \{b, f, g, i\}$. 
\item If $b_y = F$, we set $\beta'(v) = T$ for all $v \in \{b, c, e, g\}$ and $\beta'(w) = F$ for all $w \in \{a, d, f, h, i\}$. 
\item If $b_z = F$, we set $\beta'(v) = T$ for all $v \in \{c,f,h,i\}$ and $\beta'(w) = F$ for all $w \in \{a, b, d, e, g\}$. 
\end{itemize}  
It is easy to verify that these truth assignment satisfy $\mathcal{H}(\bar{x}, \bar{y}, \bar{z})$. Note that we did only specify the truth value of one variable in $\{x, y, z\}$ in each case, e.g., in the case $b_x = F$ the given assignment satisfies $\mathcal{H}(\bar{x}, \bar{y}, \bar{z})$ for any truth values $b_y, b_z$ assigned to $y$ and $z$, respectively.     
\end{proof}

With the help of the gadgets introduced in the two lemmata above we are now able to prove NP-completeness of \textsc{Monotone 3-Sat-($3,2$)}.

\begin{propo}\label{prop:Mon3SAT(2,3)}
\textsc{Monotone 3-Sat-($3,2$)} is NP-complete.
\end{propo}

\begin{proof}
By reduction from \textsc{3-Sat-(2,2)}, for which NP-hardness was established by Berman et al.~\cite[Theorem 1]{Berman2003}. Given an instance of the latter with a set $V$ of variables and a set $C$ of clauses over $V$, let $n:=|V|$. Recall that $n$ is a multiple of $3$ due to $4n = 3|C|$. For each variable $x_i \in V$, we introduce two new variables $x_{i,1}$, $x_{i,2}$ and replace the two negated appearances with $x_{i, 1}$ and the two unnegated appearances with $x_{i, 2}$, respectively. Then, we remove all negations and introduce the following clauses for $i \in \{1,2, \ldots, n\}$:
\begin{align*}
\{\{x_{i,1},\,x_{i,2}\}, \{\overline{x_{i,1}},\,\overline{x_{i,2}}, \overline{y_i}\}\} \cup \mathcal{G}(y_i, y_i, y_i) \cup \mathcal{H}(\overline{y_i}, \overline{x_{i,1}},\,\overline{x_{i,2}}),
\end{align*}
where $y_i$, $1 \leq i \leq n$, are new variables. Note that each variable appears exactly  three times unnegated and twice negated ($x_{i,1}$ and $x_{i,2}$ each appear twice unnegated in the original clauses). It is easy to see that the introduced clauses are satisfiable if and only if we assign different truth values to $x_{i,1}$ and $x_{i,2}$ for all $i \in \{1, \ldots, n\}$. Hence, the constructed formula is satisfiable if and only if the original formula is satisfiable. Now, observe that the number of positive 2-clauses is equal to $n=3k$, for some $k \in \mathbb{N}$. Next, for $j\in \{1,\ldots,k\}$, we 
 introduce the new variables $u_{j},v_{j},w_{j}$, and the clauses 
\[
\mathcal{H}(\overline{u_j}, \overline{v_j}, \overline{w_j}) \cup \mathcal{H}(\overline{u_j}, \overline{v_j}, \overline{w_j}) \cup \mathcal{G}(v_j, v_j, v_j) \cup \mathcal{G}(w_j, w_j, w_j).
\]
Note that each variable $u_j$ has the forced truth value false and appears exactly twice negated, whereas each of the variables $v_j$, $w_j$ appears exactly three times unnegated and twice negated. 

Finally, for $i=1,\ldots, n$  we replace the $2$-clause $\{x_{i,1},x_{i,2}\}$ with the $3$-clause $\{x_{i,1},x_{i,2},u_{\ell}\}$, 
where $3(\ell-1)<i\leq3\ell$. Observe that the resulting formula is indeed an instance of \textsc{Monotone 3-Sat-($3,2$)} and satisfiable if and only if the original formula is satisfiable since adding $u_j$ to the positive 2-clauses has no effect on the satisfiability. We conclude the proof by remarking that the transformation is polynomial. 
\end{proof}

\textbf{Remark.} Observe that the number of variables in each instance of \textsc{Monotone 3-Sat-($3,2$)} is divisible by 3. Now, we can increase the appearances of three variables $x, y, z$ by exactly one each by introducing the following clauses:
\begin{multicols}{5} 
\begin{enumerate}
\item $\{a, b, x\}$
\item $\{c, d, y\}$
\item $\{e, f, z\}$
\item $\{a, b, c\}$
\item $\{a, b, d\}$
\item $\{a, e, f\}$
\item $\{b, e, f\}$
\item $\{c, d, e\}$
\item $\{c, d, f\}$
\item $\{\overline{a}, \overline{b}, \overline{d}\}$
\item $\{\overline{a}, \overline{b}, \overline{f}\}$
\item $\{\overline{c}, \overline{d}, \overline{e}\}$
\item $\{\overline{c}, \overline{e}, \overline{f}\}$
\end{enumerate}
\end{multicols}

Here, $V_\text{aux} = \{a, b, c, d, e, f\}$ are new variables. Note that each introduced variable $v \in V_\text{aux}$ appears exactly four times unnegated and twice negated. Since setting all variables in $\{a, c, e\}$ true and all variables in $\{b, d, f\}$ false satisfies the above collection of clauses (i.e., the truth values of $x, y, z$ are irrelevant for the satisfiability of the introduced clauses).

With Proposition~\ref{prop:Mon3SAT(2,3)} and the above remark on how to increase the number of unnegated variable appearances in an instance of \textsc{Monotone 3-Sat-($3,2$)} we get the following corollary. 

\begin{cor}
\textsc{Monotone 3-Sat-($4,2$)} is NP-complete.
\end{cor}

Therewith, the dichotomy for \textsc{Monotone 3-Sat-E6} is set as follows. 

\begin{thm}
\textsc{Monotone 3-Sat-($p$,$q$)} with $p+q=6$ is NP-complete if $p \not \in \{0, 6\}$ and trivial otherwise. 
\end{thm}

\subsection{On a restricted variant of \textsc{Monotone 3-Sat-4}}\label{sub:a-restr}

Finally, we consider \textsc{Monotone 3-Sat-E4}, i.e., with exactly four appearances of each variable. We begin this short section with the following lemma.

\begin{lem}\label{lem:CforMon3SATE4}
Consider the following collection $\mathcal{C}(x, y)$ of monotone clauses, where $V_\text{aux} = \{a, b, \ldots, h\}$ are new variables.

\begin{multicols}{3}
\begin{enumerate}
\item $\{\bar{a},\, \bar{c},\, \bar{e}\}$
\item $\{\bar{a},\, \bar{c},\, \bar{f}\}$
\item $\{\bar{a},\, \bar{d},\, \bar{g}\}$
\item $\{\bar{b},\, \bar{c},\, \bar{h}\}$
\item $\{\bar{b},\, \bar{e},\, \bar{g}\}$
\item $\{\bar{b},\, \bar{f},\, \bar{g}\}$
\item $\{\bar{d},\, \bar{e},\, \bar{h}\}$
\item $\{\bar{d},\, \bar{f},\, \bar{h}\}$
\item $\{a,\,b,\,x\}$
\item $\{c,\,d,\,x\}$
\item $\{e,\,f,\,x\}$
\item $\{g,\,h,\,y\}$
\end{enumerate}
\end{multicols}
Then, a truth assignment $\beta$ for $\{x, y\}$ can be extended to a truth assignment $\beta'$ for $\{x, y\} \cup V_{\text{aux}}$ that satisfies $\mathcal{C}(x, y)$ if and only if $\beta(v) = T$ for at least one $v \in \{x, y\}$.
\end{lem}

\begin{proof}
We show that this collection of clauses is unsatisfiable if $x$ and $y$ are both set false. By clause 11 at least one of $e, \,f$ has to be set true. As a consequence clauses 1,\,2; 5,\,6 and 7,\,8 imply that the additional clauses $i.\;\{\bar{a},\,\bar{c}\}$; $ii.\;\{\bar{b},\,\bar{g}\}$ and $iii.\;\{\bar{d},\,\bar{h}\}$ would have to be satisfied as well. First consider any variable assignment $\beta$ with $\beta (g) = F$:
\[
\beta (g) = F \overset{12.}{\Rightarrow} \beta(h) = T \overset{iii.}{\Rightarrow} \beta(d) = F \overset{10.}{\Rightarrow} \beta(c) = T \overset{i.}{\Rightarrow} \beta(a) = F \overset{9.}{\Rightarrow} \beta(b) = T. 
\]
Thus, clause 4 is not satisfied. Now we consider the other case $\beta(g) = T$:
\[
\beta (g) = T \overset{ii.}{\Rightarrow} \beta(b) = F \overset{9.}{\Rightarrow} \beta(a) = T \overset{i.}{\Rightarrow} \beta(c) = F \overset{10.}{\Rightarrow} \beta(d) = T. 
\]
Thus, clause 3 is not satisfied. Consequently, the collection of clauses is unsatisfiable if $\beta(x) = \beta(y) = F$. Without clause 12, there is a satisfying truth assignment: Set all variables in $\{a,\,g,\,h\}$ false and all variables in $\{b, c, d, e, f\}$ true. Hence, the collection of clauses is satisfiable if $\beta(y) = T$. Finally, if $\beta(x) = T$, we can satisfy all clauses by setting $h$ true and all variables in $\{a, b, \ldots, g\}$ false. 
\end{proof}

Lemma~\ref{lem:CforMon3SATE4} implies the following corollary, where $\mathcal{C}(\cdot,\cdot)$ refers to the respective set of clauses introduced in Lemma~\ref{lem:CforMon3SATE4}.

\begin{cor}\label{cor:B(x,y,z)}
Consider the collection of clauses $\mathcal{B}(x, y, z) = \mathcal{C}(u, x) \cup \mathcal{C}(v, y) \cup \mathcal{C}(w, z) \cup \{\{\bar{u},\,\bar{v},\,\bar{w}\}\}$, and let $V$ be its associated set of variables. Then, a truth assignment $\beta$ for $\{x, y, z\}$ can be extended to a truth assignment $\beta'$ for $V$ that satisfies $\mathcal{B}(x, y, z)$ if and only if $\beta(v) = T$ for at least one $v \in \{x, y, z\}$. 
\end{cor}

\begin{cor}\label{cor:barB(barx,bary,barz)}
Consider the collection of clauses $\bar{\mathcal{B}}(\bar{x}, \bar{y}, \bar{z})$ obtained from $\mathcal{B}(x, y, z)$ by replacing each literal with its negation, and let $V$ be its associated set of variables. Then, a truth assignment $\beta$ for $\{x, y, z\}$ can be extended to a truth assignment $\beta'$ for $V$ that satisfies $\bar{\mathcal{B}}(\bar{x}, \bar{y}, \bar{z})$ if and only if $\beta(v) = F$ for at least one $v \in \{x, y, z\}$.
\end{cor}

\textbf{Remark.} Each instance of an gadget $\mathcal{B}(x, y, z)$ (resp. $\bar{\mathcal{B}}(\bar{x}, \bar{y}, \bar{z})$) has its own new auxiliary variables (i.e., the variables that are not in $\{x, y, z\}$).\\

The above corollaries will be useful to prove that \textsc{Monotone 3-Sat-E4} is NP-complete even when restricted to instances in which each variable appears either three times unnegated and once negated or three times negated and once unnegated.

\begin{thm}\label{thm:Mon3SATE4}
\textsc{Monotone 3-Sat-E4} is NP-complete even if each variable appears three times unnegated and once negated or three times negated and once unnegated. 
\end{thm}

\begin{proof}
We show NP-hardness by reducing from \textsc{3-Sat-(2,2)}, for which NP-hardness was established by Berman et al.~\cite[Theorem 1]{Berman2003}. Given an instance~$\mathcal{I}$ of the latter with a set $V$ of variables and a set $C$ of clauses over $V$, let $n:=|V|$. For each variable $x_i \in V$, we introduce two new variables $x_{i,1}$, $x_{i,2}$ and replace the two negated appearances with $x_{i, 1}$ and the two unnegated appearances with $x_{i, 2}$, respectively. Then, we remove all negations and introduce the following clauses for $i \in \{1,2, \ldots, n\}$, where $z_i$ and $y_i$ are new variables:
\[
\{\{x_{i,1},\,x_{i,2}, y_i\}, \{\overline{x_{i,1}},\,\overline{x_{i,2}}, \overline{z_i}\}\} \cup \bar{\mathcal{B}}(\overline{y_i},\overline{y_i}, \overline{y_i} ) \cup \mathcal{B}(z_i, z_i, z_i).
\]   
Let $V_i$ denote the variables appearing in the clauses introduced above. By construction and corollaries~\ref{cor:B(x,y,z)} and~\ref{cor:barB(barx,bary,barz)}, a truth assignment $\beta_i$ for $\{x_{i,1},\,x_{i,2}\}$ can be extended to a truth assignment $\beta'_i$ for $V_i$ that satisfies these clauses if and only if $\beta_i(x_{i,1}) \neq \beta_i(x_{i,2})$. Now it is straightforward to verify that the constructed set of clauses is satisfiable if and only if the given instance $\mathcal{I}$ is satisfiable.  

By construction of $\bar{\mathcal{B}}(\overline{y_i},\overline{y_i}, \overline{y_i})$ and $\mathcal{B}(z_i, z_i, z_i)$, each variable in $\bigcup_{i=1}^n V_i$ appears  three times unnegated and once negated or three times negated and once unnegated. Moreover, each variable in $\bigcup_{i=1}^n \{x_{i,1}, x_{i,2}\}$ appears once unnegated and once negated in the introduced clauses, and twice unnegated in the original clause set. Also observe that, by construction, all clauses are monotone. Hence, we constructed an instance of \textsc{Monotone 3-Sat-E4} where each variable appears three times negated and once unnegated or three times unnegated and once negated.

We conclude the proof by remarking that the transformation is polynomial.~\end{proof}

Finally, dropping the monotonicity condition we point out that Theorem~\ref{thm:Mon3SATE4} implies also an interesting hardness result for \textsc{3-Sat-E4}. For instance, replacing each variable $x$ that appears negated exactly three times and unnegated exactly once with a new variable $z$ such that literal $z$ replaces literal $\bar{x}$ and literal $\bar{z}$ replaces literal $x$, it follows that \textsc{3-Sat-E4} is NP-complete even if each variable appears exactly once negated and exactly three times unnegated. An analogous result  follows for the case  that each variable appears exactly three times  negated and exactly once unnegated. Therewith, we complement a result by Berman et al.~\cite{Berman2003} stating that \textsc{3-Sat-E4} is NP-complete even if each variable appears exactly twice unnegated and exactly twice negated. We summarize these findings in terms of the corollary below (for the sake of completeness, we include also the result by Berman et al.~\cite[Theorem 1]{Berman2003}). 

\begin{cor}
\textsc{3-Sat-E4} is NP-complete even if either
\begin{itemize}
\item each variable appears exactly three times unnegated and once negated, or 
\item each variable appears exactly three times negated and once unnegated, or
\item each variable appears exactly twice unnegated and twice negated~\cite{Berman2003}, respectively. 
\end{itemize}
\end{cor}

\section{Conclusion}\label{sec:con}

We have shown that \textsc{Not-All-Equal 3-Sat} remains NP-complete for linear and monotone formulas in CNF, where each clause contains exactly 3 distinct variables and every variable appears in exactly 4 clauses. In a sense, these parameters establish a sharp separation line between polynomial time solvability and \textsc{NP}-completeness, since it is known that \textsc{Not-All-Equal 3-Sat} can be decided in polynomial time if
\begin{itemize}
\item the formula is exact linear~\cite[Corollary 2]{porschen09} (i.e., each pair of distinct clauses shares exactly one variable),  
\item each clause contains at most 2 distinct variables~\cite[Theorem 1]{porschen05}, or
\item each clause is monotone and contains exactly 3 distinct variables, and each variable appears exactly (at most) 3 times~\cite[Theorem 4]{porschen04}, respectively.  
\end{itemize}
Further, we provided NP-completeness of \textsc{Monotone 3-Sat-$(k, k)$} for all $k \geq 3$. By a result of Tovey~\cite[Theorem 2.4]{tovey84} the problem is trivial for $k=1$, i.e., all such instances are satisfiable. For the remaining case $k = 2$ we were able to show that it is either trivial or NP-complete, and that NP-completeness holds if the three literals in each clause are not required to be distinct. Hence, we present the following challenge for future research in order to clarify the complexity status for $k=2$:\\

\noindent{\bf Challenge 1.} Find an unsatisfiable instance of \textsc{Monotone 3-Sat-$(2, 2)$} or prove that all instances are satisfiable.\\   

\noindent
Another focus of the paper was laid on \textsc{Monotone 3-Sat-$(k,1)$}, where each variable appears exactly $k$ times unnegated and once negated respectively. For this variant, we proved NP-completeness for all $k \geq 5$. Again, by Tovey~\cite[Theorem 2.4]{tovey84} the problem is trivial for $k \leq 2$. The cases $k=3$ and $k=4$ are, to the best of our knowledge, open; we hence state the following second challenge for future research: \\

\noindent{\bf Challenge 2.} Is \textsc{Monotone 3-Sat-$(k,1)$} NP-hard for $k \in \{3, 4\}$?\\   

\bibliographystyle{alpha}
\bibliography{mylit}

\end{document}